\documentclass[format=acmsmall, review=false]{acmart}
\usepackage{acm-ec-26}
\usepackage{booktabs} 
\usepackage[ruled]{algorithm2e} 
\usepackage{natbib} 

\newtheorem{remark}{Remark}
\newtheorem{assumption}{Assumption}
\SetAlFnt{\small}
\SetAlCapFnt{\small}
\SetAlCapNameFnt{\small}
\SetAlCapHSkip{0pt}
\IncMargin{-\parindent}

\setcitestyle{authoryear}

\title[Pricing Time, Not Just Tokens: Latency-Aware Mechanism Design for LLM Inference]{Pricing Time, Not Just Tokens: Latency-Aware Mechanism Design for LLM Inference}

\author{Ian McDougall}
\email{imcdougall@wisc.edu}
\orcid{0009-0005-4339-7233}
\affiliation{%
  \institution{University of Wisconsin-Madison}
  \city{Madison}
  \state{WI}
  \country{USA}
}
\author{Karthikeyan Sankaralingam}
\email{karu@cs.wisc.edu}
\orcid{0000-0002-8315-2389}
\affiliation{%
  \institution{University of Wisconsin-Madison}
  \city{Madison}
  \state{WI}
  \country{USA}
}
\additionalaffiliation{%
  \institution{NVIDIA}
  \city{Santa Clara}
  \state{CA}
  \country{USA}
}

\begin{abstract}
The economic theory of LLM pricing treats tokens as a homogeneous commodity considering aggregate token count as the main features buyers and sellers consider. We model inference as a service market where buyers have three-dimensional private information---willingness-to-pay, task volume, and time preference---and utility depends on \emph{latency slack} alongside token quantities.

Our main result is a \emph{separation theorem}: discrete hardware tiers induce endogenous self-selection on time preferences, reducing three-dimensional screening to standard one-dimensional screening within each tier. We derive the cost structure from GPU inference physics---compute-bound prefill and bandwidth-bound decode---and characterize optimal tiered mechanisms via virtual-value techniques. Optimal per-task prices are volume-independent, providing theoretical grounding for flat per-token API pricing.

We verify the mechanism empirically by calibrating to 8-GPU clusters of H100 and B200 hardware. The separation theorem holds in 83\% of 105 tested configurations overall, rising to 96\% at economically relevant WTP scales. A seller adopting two-tier pricing under the optimal mechanism captures 26--66\% higher profit than the best single-tier alternative, with gains driven by efficient cross-tier allocation in regimes where hardware costs are a significant fraction of per-request value.
\end{abstract}

\begin{document}

\maketitle

\section{Introduction}

The prevailing economic literature on Large Language Models (LLMs) treats inference as a spot market for tokens---a homogeneous commodity priced by volume with no temporal dimension. This framing, exemplified by \citet{bergemann2025economics}, assumes the buyer's production function depends solely on the total quantity of tokens consumed. Consequently, the singular metric of economic value is aggregate token count, and the question of \emph{when} those tokens are delivered plays no role in either buyer utility or seller cost.

We complement this throughput-dominated framework by modeling LLM inference as a \emph{service market with delivery time commitments}. The distinction is economically substantive: in a spot market for tokens, the seller's problem reduces to cost minimization over a fungible input; in a service market, the seller must jointly optimize over price, quality, and congestion. This reframing brings LLM pricing into dialogue with the classical economics of service industries---airlines, telecommunications, cloud computing---where time-sensitivity induces vertical differentiation and screening on latency preferences.

The operational reality of commercial inference supports this reframing. Cloud providers have broken the flat-rate paradigm: AWS Bedrock offers ``Priority Tier'' pricing guaranteeing faster generation, while Google Vertex AI and Azure OpenAI sell ``Provisioned Throughput'' capacity.\footnote{The premium on inference speed is substantial: a recent \$20 billion Nvidia--Groq licensing agreement \citep{nellis_2025} is predicated on Groq's \emph{architecturally specialized} inference hardware, which achieves latency reductions~\citep{Groq2024InferenceSpeed} through fundamentally different memory design (SRAM-based, eliminating the HBM bandwidth bottleneck) rather than incremental GPU generational improvement. Claude also announced the release of their Opus 4.6 LLM with ``Fast Mode'', which improves latency without compromising quality, indicating similar trends in the LLM design space \citep{claude26}.}

\begin{figure}[tbp]
    \centering
    \includegraphics[width=0.9\linewidth]{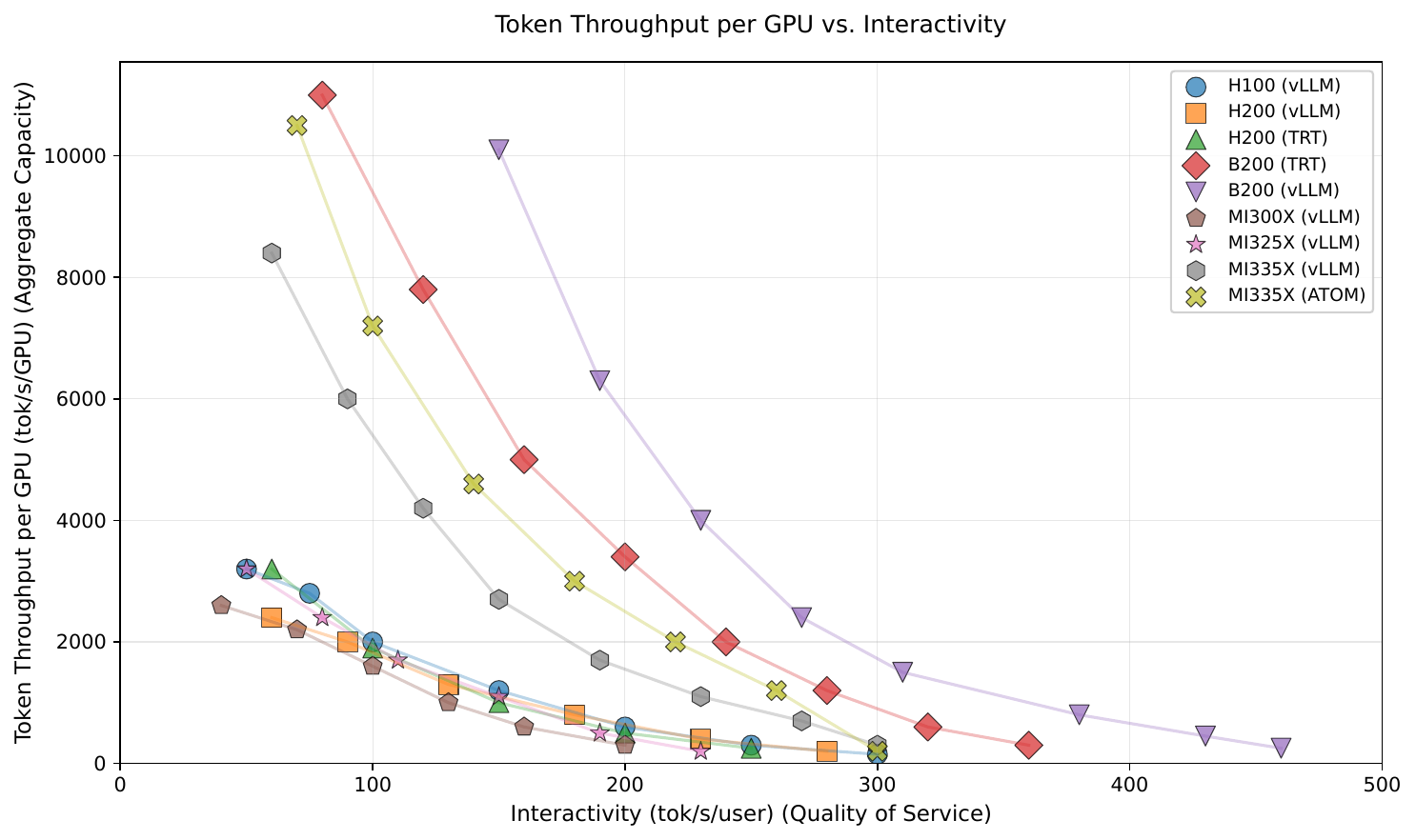}
    \caption{\textbf{The throughput--latency production frontier.} Each curve traces feasible combinations for a GPU architecture: per-user quality (tokens/second/user) versus aggregate utilization (tokens/second/GPU). The negative slope reflects capacity constraints: faster service requires reserving resources that could otherwise be pooled. Data from \citet{Chen_Patel_Nishball_Quilici_Wen_2025}.}
    \label{fig:throughput-vs-latency}
\end{figure}

Figure~\ref{fig:throughput-vs-latency} illustrates this frontier. For GPU-based inference, a fundamental tension exists between throughput (aggregate tokens processed) and latency (delivery speed). The seller faces a choice between returns to specialization (faster service per buyer) and returns to scale (more buyers served via statistical multiplexing). In the regime studied by \citet{bergemann2025economics}, returns to scale dominate. Market evidence suggests latency-dominated regimes are also worth exploring: newer generations of GPU hardware demand premiums despite comparable throughput, and chip designers invest heavily in high-bandwidth interconnects for strong scaling effects. Observing throughput-maximization markets alone cannot explain why buyers pay for performance or why tiered offerings proliferate.

\subsection{Primary Contributions}

We develop a mechanism design framework where a monopolistic seller offers LLM inference services to buyers with three-dimensional private information: willingness-to-pay $w$, task volume $s$, and time preference $\gamma$. The buyer's utility function incorporates \emph{latency slack}---the gap between maximally tolerable and actual delivery time---as a first-order argument alongside token quantities. This specification captures the temporal dimension that volume-centric models omit.

Our first contribution derives an empirically-grounded cost structure from the production technology of GPU-based inference. We decompose latency into compute-bound (prefill) and memory-bound (decode) phases, yielding a cost function that is strictly convex in service quality, with marginal cost diverging as latency approaches hardware limits. This convexity reflects the fundamental physics: achieving low latency requires sacrificing batching efficiency, and the tradeoff becomes increasingly severe near capacity constraints. The cost structure satisfies the technical conditions required for Mussa-Rosen screening while remaining empirically calibrated to measured hardware performance.

Our second contribution establishes a \emph{separation theorem} that reduces the intractable three-dimensional screening problem to a tractable one. We prove that when hardware tiers differ sufficiently in latency capabilities, a single-crossing property in time preferences induces endogenous self-selection: high-$\gamma$ (time-sensitive) buyers strictly prefer premium tiers, while low-$\gamma$ (time-tolerant) buyers select economy tiers. The key technical condition is that the quality gain from upgrading tiers is strictly increasing in $\gamma$---a property we verify holds across empirically calibrated parameters. This separation transforms the problem into standard one-dimensional screening on willingness-to-pay within each tier, with the time-preference dimension screened through tier choice at zero informational cost to the seller.

Our third contribution characterizes the optimal mechanism using virtual-value techniques adapted to our convex cost structure. Within each tier, quality schedules satisfy the Mussa-Rosen first-order condition equating virtual value with marginal cost; buyers below a threshold are excluded. The menu exhibits tiered pricing that mirrors observed practice: premium tiers charge for the \emph{right to not batch aggressively}---valuable precisely to time-sensitive buyers. We establish that optimal per-task prices are independent of volume: the cost structure $C(Q,s) = s \cdot \tilde{c}(Q/s)$ implies no quantity discounts, providing theoretical grounding for the flat per-token pricing dominant in current API markets.

We extend the analysis to continuous $\gamma$ distributions, where $K$ discrete hardware tiers induce an endogenous partition of the time-preference space. Numerical analysis reveals a sharp regime dependence: when the premium tier is a generational upgrade---uniformly faster at a modest cost premium---it dominates all buyer types at high per-request values, and single-tier pricing suffices. When hardware costs are a significant fraction of per-request value, our tiered mechanism delivers profit gains of 26--66\%, because efficient cross-tier allocation has proportionally larger impact on margins. As LLM inference matures from a high-margin novelty into a ubiquitous utility, hardware costs consume an increasing share of per-request value, making tiered pricing increasingly decisive.

\subsection{Related Work}
\label{sec:related}

\paragraph{Economics of AI Services.} \citet{bergemann2025economics} formalize token allocation for fine-tuning and inference but target throughput-maximizing regimes where timing is secondary; \citet{dutting2025mechanism} price generative-AI content rather than the inference service. We complement both by isolating the regime where \textit{latency heterogeneity} governs utility, introducing a three-dimensional screen over willingness-to-pay, volume, and time preference. Broader AI economics \citep{korinek2024ai, acemoglu2024learning, agrawal2024economic} does not model inference pricing.

\paragraph{Versioning and Multidimensional Screening.} \citet{varian2000versioning} shows a ``damaged'' (higher-latency) good enables segmentation when marginal costs are near zero; here latency arises endogenously from hardware and batching rather than being injected, paralleling the ``Paris Metro Pricing'' of \citet{odlyzko1999paris}. Multidimensional screening is generally intractable, requiring complex ironing \citep{rochet1998ironing, armstrong1996multiproduct}; our separation theorem circumvents this, reducing the problem to one-dimensional Mussa-Rosen screening within each tier \citep{mussa1978monopoly, myerson1981optimal}. Our volume-independent per-task prices contrast with the quantity discounts of standard nonlinear pricing \citep{maskin1984monopoly}, since the marginal cost of latency slack scales linearly with volume.

\paragraph{Systems Foundations.} Our cost structure derives from the bifurcation of GPU inference into compute-bound prefill and bandwidth-bound decode, building on the continuous batching of Orca \citep{yu2022orca} and PagedAttention of vLLM \citep{kwon2023efficient}. Although we characterize tiers via monolithic GPU performance (H100 vs.\ B200), the framework extends to prefill-decode disaggregation \citep{patel2024splitwise}, where the premium tier abstracts a low-latency decode pool and the economy tier a throughput-optimized one. Specialized accelerators such as TPUs \citep{jouppi2017datacenter} shift the latency coefficients but reinforce the batching-versus-slack tradeoff driving separation.

\paragraph{Service Operations, SLOs, and Timing.} Emerging reasoning models (OpenAI o1, DeepSeek R1) push computation to test-time, producing latency-insensitive decode; our mechanism captures exactly this regime, where hardware costs are significant but latency sensitivity is low. We position the work as the economic dual of the Service Level Objective literature: Clockwork \citep{gujarati2020clockwork} and Shepherd \citep{zhang2023shepherd} treat latency as a hard constraint, whereas we internalize it into the objective, exposing the extractable value of the slack that aggressive scheduling discards. This generalizes the ``FedEx problem'' of \citet{fiat2016fedex} to continuous token streaming, replacing discrete deadlines with smooth $\gamma$ and abstracting from queueing dynamics \citep{mendelson1990optimal, hassin2003queue} to isolate the structural signaling of speed \citep{debo2005signaling}. Engineering foundations draw on \citet{MaPatterson2026LLMInferenceHardware} and \citet{davies2025liminal}.

\section{Model}
\label{sec:model}

\subsection{Environment}

We study a market for LLM inference services between a monopolist seller and a unit measure of buyers with heterogeneous preferences. Each buyer privately learns their type $\theta = (w, s, \gamma)$---willingness-to-pay, task volume, and time preference---then selects from a menu of contracts posted by the seller. Production technology is common knowledge; all parties have quasilinear preferences.

A buyer of type $\theta = (w, s, \gamma)$ arrives with $s$ tasks requiring LLM processing, indexed by $i \in \{1, \ldots, s\}$. For each task $i$, the buyer selects input tokens $x_i \geq 0$ and output tokens $y_i \geq 0$. Tasks are processed in batches of size $B \in \{1, \ldots, s\}$ on hardware tier $g \in \mathcal{G}$, where $\mathcal{G}$ denotes the set of available tiers ordered by capability. The constraint $B \leq s$ reflects that batch size cannot exceed the number of available tasks.

\subsection{Buyer Preferences}
\label{sec:buyer-preferences}

The value a buyer derives from task $i$ depends on both the tokens allocated and the speed of delivery. We specify preferences through a production function over inputs and latency:
\begin{equation}
    v(x, y; L, \gamma) = x^{\alpha} y^{\beta} (L_{\max} - L)^{\gamma},
    \label{eq:production}
\end{equation}
where $\alpha, \beta > 0$ are common parameters governing returns to token inputs, $\gamma > 0$ is the buyer's time-preference intensity, and $L_{\max}$ is the buyer's maximum tolerable latency. We impose $\alpha + \beta + \gamma < 1$ for all $\gamma$ in the support, ensuring diminishing returns to the composite input bundle.

This Cobb-Douglas specification with a deadline constraint follows the literature on time-sensitive service systems \citep{fiat2016fedex}. The term $(L_{\max} - L)$ captures \emph{latency slack}---the margin between deadline and delivery. The exponent $\gamma$ governs time-preference intensity: high-$\gamma$ buyers experience rapid utility loss as slack diminishes (real-time applications), while low-$\gamma$ buyers tolerate delay (batch workloads).

The total utility for a buyer of type $\theta = (w, s, \gamma)$ who receives allocations $(x_i, y_i)_{i=1}^s$, experiences latency $L$, and pays transfer $T$ is:
\begin{equation}
    u(\theta) = w \sum_{i=1}^{s} v(x_i, y_i; L, \gamma) - T.
    \label{eq:utility}
\end{equation}
Since $v$ is strictly concave and all tasks share latency $L$, the optimal allocation is symmetric ($x_i = x$, $y_i = y$ for all $i$), so we work with per-task allocations without loss of generality \citep{bergemann2025economics}.

\subsection{Production Technology: Latency}
\label{sec:latency}

Unlike tokens, latency is a batch-level phenomenon: all tasks in a batch experience common delay $L = L(B, x, y, g)$, where $B$ is batch size and $g \in \mathcal{G}$ is the hardware tier.

The latency function captures two sources of delay: (i)~\emph{setup cost} (prefill phase)---the time to initialize computation, scaling superlinearly in input complexity due to $O(n^2)$ attention; and (ii)~\emph{marginal service time} (decode phase)---per-token processing that increases with batch size due to congestion externalities as co-batched requests compete for memory bandwidth.

Our empirical analysis (Section~\ref{sec:latency-model}) yields a latency function of the form:
\begin{equation}
    L(B, x, y, g) = \underbrace{a_g x^2 + b_g x + c_g}_{\text{prefill (setup cost)}} + \underbrace{y(d_g + e_g B + f_g Bx)}_{\text{decode (marginal service time)}},
    \label{eq:latency}
\end{equation}
where $(a_g, b_g, c_g, d_g, e_g, f_g)$ are hardware-specific parameters estimated from measured GPU latencies.\footnote{For notational simplicity, we suppress the $g$ subscript and write $L(B, x, y)$ when the hardware tier is clear from context.}

The prefill phase does not depend on batch size $B$: modern continuous-batching systems process prefill requests individually \cite{kwon2023efficient}. The decode phase is memory-bandwidth-bound, with terms $e_g B$ and $f_g Bx$ capturing congestion as batched sequences compete for bandwidth. This formulation embodies the fundamental tension: larger batches amortize setup costs but increase latency through congestion, with slack $L_{\max} - L$ decreasing in $(B, x, y)$.

\subsection{Type Space and Distributional Assumptions}
\label{sec:types}

Buyer heterogeneity spans three dimensions: willingness-to-pay $w \in \mathbb{R}_+$, task volume $s \in \mathbb{N}^+$, and time preference $\gamma \in \Gamma \subseteq (0, 1 - \alpha - \beta)$. The full type space is $\Theta = \mathbb{R}_+ \times \mathbb{N}^+ \times \Gamma$.

\begin{assumption}[Independence]
\label{ass:independence}
The type components $w$, $s$, and $\gamma$ are independently distributed according to $F_w$, $F_s$, and $F_\gamma$ respectively.
\end{assumption}

\begin{assumption}[Regularity]
\label{ass:regularity}
The distribution $F_w$ admits a continuous density $f_w$ on $\mathbb{R}_+$ with monotone hazard rate:
\begin{equation}
    \frac{d}{dw} \left[ \frac{1 - F_w(w)}{f_w(w)} \right] \leq 0 \quad \text{for all } w > 0.
\end{equation}
\end{assumption}

Assumption~\ref{ass:independence} permits analysis of each screening dimension in partial isolation. Assumption~\ref{ass:regularity} ensures that the virtual value $\varphi(w) \equiv w - (1 - F_w(w))/f_w(w)$ is strictly increasing \citep{myerson1981optimal}.

Our baseline assumes binary time preference: $\gamma \in \{\gamma_L, \gamma_H\}$ with $0 < \gamma_L < \gamma_H < 1 - \alpha - \beta$, capturing the bimodality between latency-critical and latency-tolerant workloads. Section~\ref{sec:continuous} extends to continuous $\gamma$. Our main result (Theorem~\ref{thm:separation}) establishes that the optimal mechanism induces self-selection: high-$\gamma$ buyers prefer the premium tier, low-$\gamma$ buyers prefer economy, reducing the three-dimensional problem to one-dimensional screening on $w$ within each tier.

\subsection{Seller's Cost Structure}
\label{sec:cost-structure}

The seller incurs cost $c_g \cdot L$ for running a batch of duration $L$ on tier $g$. Processing $s$ tasks in batches of size $B$ yields total cost:
\begin{equation}
    \text{Cost}(s, B, x, y, g) = \frac{s}{B} \cdot c_g \cdot L(B, x, y, g).
    \label{eq:total-cost}
\end{equation}
The ratio $L/B$ is per-task machine time, which is non-monotonic in $B$: larger batches reduce setups but increase per-batch latency, driving the fundamental tradeoff between batching efficiency and service speed.

The \emph{per-task cost function} is the minimum cost to deliver quality $q$ on tier $g$:
\begin{equation}
    \tilde{c}^*(q; g, \gamma) = \min_{x, y, B} \left\{ \frac{c_g \cdot L(B, x, y, g)}{B} \;\middle|\; v(x, y; L, \gamma) = q, \; L < L_{\max} \right\}.
    \label{eq:cost-function}
\end{equation}
After $\gamma$-separation (Section~\ref{sec:two-types}), each tier serves a single $\gamma$ type, permitting notation $\tilde{c}^*(q; g)$. Theorem~\ref{thm:cost-convexity} establishes that $\tilde{c}^*(q; g)$ is strictly convex with $\tilde{c}^{*\prime}(0) = 0$ and $\tilde{c}^{*\prime}(q) \to \infty$ as $q \to \bar{q}(g)$, ensuring the Mussa-Rosen approach yields well-behaved solutions.

For $s$ tasks at aggregate quality $Q = s \cdot q$, the total cost function takes the form:
\begin{equation}
    C(Q, s; g) = s \cdot \tilde{c}^*(Q/s; g).
    \label{eq:aggregate-cost}
\end{equation}
This structure implies that marginal cost depends only on per-task quality $q = Q/s$, not on scale $s$ directly---a property with important implications for quantity discounts.
\section{Optimal Menu Pricing}
\label{sec:optimal-menu}

Having established the primitives of buyer preferences and seller costs, we turn to the mechanism design problem. A monopolist offering LLM inference services must screen buyers who differ along three dimensions: willingness-to-pay $w$, task volume $s$, and time preference $\gamma$. Unlike standard token-pricing models, our setting requires the seller to jointly determine computational resources---not merely quantities---since identical tokens delivered faster command higher value from time-sensitive buyers.

Formally, the seller commits to a direct mechanism assigning, for each reported type $\theta = (w, s, \gamma)$, a complete specification of service delivery:
\begin{equation}
    \left( (x_i(\theta), y_i(\theta))_{i \in \mathbb{N}^+}, B(\theta), g(\theta), T(\theta) \right)_{\theta \in \Theta}.
    \label{eq:mechanism}
\end{equation}
This includes per-task token allocations $(x_i, y_i)$, batch size $B$, hardware tier $g \in \mathcal{G}$, and monetary transfer $T$. Allocations are contractible: once a buyer selects a menu option, reallocation of tokens across tasks is precluded. Payoffs are quasilinear in the transfer.

Multi-dimensional screening is notoriously difficult \citep{armstrong1999multi}. However, the economics of inference---specifically, the discrete nature of hardware and the structure of latency costs---admits a clean decomposition. First, discrete hardware tiers induce self-selection on $\gamma$: high-$\gamma$ buyers strictly prefer premium (low-latency) hardware while low-$\gamma$ buyers prefer economy tiers, revealing time preference without explicit screening (Theorem~\ref{thm:separation}). In the direct mechanism~\eqref{eq:mechanism}, no buyer benefits from misreporting $\gamma$, as the seller's optimal tier assignment coincides with each buyer's preferred tier. Second, within each tier the problem reduces to classical one-dimensional screening on $w$, amenable to Mussa-Rosen techniques \citep{mussa1978monopoly}.

We develop these results for binary time-preference types (Section~\ref{sec:two-types}), which yields sharp closed-form characterizations. Section~\ref{sec:continuous} extends the analysis to continuous $\gamma$ distributions, introducing bunching across tiers. First, we briefly describe the empirical latency function underlying our cost structure.

\subsection{Latency Function Estimation}
\label{sec:latency-model}

We estimate the latency function~\eqref{eq:latency} from measured GPU latency data. The specification decomposes inference into compute-bound prefill and bandwidth-bound decode phases, with hardware-specific parameters $(a_g, b_g, c_g, d_g, e_g, f_g)$ estimated via nonlinear regression. Both fits achieve $R^2 > 0.99$ across the operational range (batch sizes 1--64, input contexts 1K--8K tokens), indicating the parametric form captures the essential structure of inference latency. Details of the simulation methodology and estimated coefficients appear in Appendix~\ref{app:empirical}.

\subsection{Binary Time-Preference Types}
\label{sec:two-types}

We begin with the case where time preference takes one of two values: $\gamma \in \{\gamma_L, \gamma_H\}$ with $0 < \gamma_L < \gamma_H < 1 - \alpha - \beta$. Let $\gamma = \gamma_H$ with probability $\pi_H$ and $\gamma = \gamma_L$ with probability $\pi_L = 1 - \pi_H$. This binary specification captures the pronounced bimodality observed in practice---time-critical applications (real-time agents, interactive assistants) versus time-tolerant workloads (batch processing, offline analysis)---while maintaining analytical tractability.

The seller faces a three-dimensional screening problem over $(w, s, \gamma)$. We establish that discrete hardware tiers enable \emph{endogenous separation} of the $\gamma$ dimension, reducing the problem to one-dimensional screening on willingness-to-pay $w$ within each tier.

\subsubsection{Hardware Tiers and Latency Slack}

The seller offers two hardware tiers: a \emph{premium tier} $g_H$ and an \emph{economy tier} $g_L$. These tiers differ in their cost structure and latency characteristics:
\begin{itemize}
    \item \textbf{Premium tier} $g_H$: Higher per-unit-time cost $c_{g_H}$, but capable of achieving lower latency $L_{g_H}$ for any configuration $(B, x, y)$.
    
    \item \textbf{Economy tier} $g_L$: Lower per-unit-time cost $c_{g_L} < c_{g_H}$, but with higher latency $L_{g_L} > L_{g_H}$.
\end{itemize}

For a buyer of type $\gamma$ selecting tier $g$, the optimal within-tier configuration $(x^*, y^*, B^*)(\gamma, g)$ yields an equilibrium latency slack $\mu_g$. The premium tier achieves greater slack: $\mu_{g_H}(\gamma) > \mu_{g_L}(\gamma)$ for all $\gamma$.

\subsubsection{The Separation Condition}

The key to reducing the multi-dimensional problem lies in establishing that $\gamma$ types self-select into appropriate tiers without explicit screening. The following condition is sufficient for separation.

\begin{assumption}[Quality-Difference Separation]
\label{ass:separation}
Let $q_g(\gamma)$ denote the per-task quality achieved by type $\gamma$ on tier $g$ at the cost-minimizing allocation. The quality gain from upgrading to the premium tier is strictly increasing in time preference:
\begin{equation}
    q_{g_H}(\gamma_H) - q_{g_L}(\gamma_H) > q_{g_H}(\gamma_L) - q_{g_L}(\gamma_L).
    \label{eq:separation-condition}
\end{equation}
\end{assumption}

This condition states that high-$\gamma$ buyers benefit more from premium hardware than low-$\gamma$ buyers---not merely in proportional terms, but in absolute quality units. When this holds, there exist tier prices that induce perfect self-selection: each $\gamma$ type strictly prefers the tier designed for it.

We now state the main separation result.

\begin{theorem}[$\gamma$-Separation]
\label{thm:separation}
Under Assumption~\ref{ass:separation}, for each $w > 0$ there exist tier transfers $(T_{g_H}, T_{g_L})$ such that:
\begin{enumerate}
    \item Type $\gamma_H$ strictly prefers the premium tier $g_H$;
    \item Type $\gamma_L$ strictly prefers the economy tier $g_L$;
    \item No type benefits from deviating to a tier designed for another $\gamma$ type.
\end{enumerate}
\end{theorem}

\begin{proof}
Condition~\eqref{eq:separation-condition} makes the separating-price interval non-empty for every $w>0$; the full construction is in Appendix~\ref{app:separation}.
\end{proof}

\paragraph{Economic intuition.} Two complementary forces explain why condition~\eqref{eq:separation-condition} is economically natural:

\emph{Log-supermodularity of preferences.} The production function $v = x^\alpha y^\beta \mu^\gamma$ is log-supermodular in $(\mu, \gamma)$: $\partial^2 \log v / \partial \mu \, \partial \gamma = 1/\mu > 0$. Economically, $\gamma$ is the elasticity of quality with respect to latency slack---a 1\% increase in slack yields a $\gamma$\% increase in quality. High-$\gamma$ buyers (interactive agents, real-time applications) thus value slack disproportionately, while low-$\gamma$ buyers (batch processing, offline analysis) are relatively insensitive to latency improvements.

\emph{Diminishing speedup.} Premium hardware exhibits diminishing speedup: the latency ratio decreases with batch size $B$. For example, at $B=1$, the premium tier may be $1.5\times$ faster; at $B=64$, only $1.2\times$ faster. This asymmetry arises because premium hardware's advantage concentrates in the serial floor $F$ (prefill-dominated), which is amortized away at large batch sizes. Since high-$\gamma$ buyers optimally choose small batches to maximize slack, they operate precisely in the regime where premium hardware's advantage is largest. Low-$\gamma$ buyers batch aggressively and thus experience a smaller premium-economy gap.
\footnote{Tier differentiation can also arise from \emph{scale-up domain architecture}. Appendix~\ref{app:scaleup} discusses this in detail.}

These forces combine: high-$\gamma$ buyers value slack more (log-supermodularity) \emph{and} receive more of it on premium hardware (diminishing speedup), producing the quality-difference separation that condition~\eqref{eq:separation-condition} requires. For our H100/B200 calibration, Section~\ref{sec:numerical} confirms condition~\eqref{eq:separation-condition} holds across 83\% of tested parameter configurations (96\% at economically relevant WTP scales), with the quality-difference gain from upgrading to B200 strictly increasing in $\gamma$ throughout the operational range.

\subsubsection{Cost Function Convexity}
\label{sec:cost-convexity}

Having established separation, we prove that the per-task cost function is strictly convex, enabling standard screening techniques within each tier. We first establish uniqueness of the cost-minimizing allocation.

For clarity, define the auxiliary functions:
\begin{align}
    D(x, y) &\equiv y(e + fx), \label{eq:D-def} \\
    F(x, y) &\equiv ax^2 + bx + c + yd. \label{eq:F-def}
\end{align}
The latency function~\eqref{eq:latency} can then be written as $L = B \cdot D(x, y) + F(x, y)$, which is affine in batch size $B$.

\begin{lemma}[Uniqueness of Cost-Minimizing Allocation]
\label{lem:uniqueness}
For any quality level $q \in (0, \bar{q})$, the cost-minimizing input bundle $(x^*, y^*, B^*)$ solving~\eqref{eq:cost-function} is unique.
\end{lemma}

\begin{proof}
The production function $v(x, y; L, \gamma)$ with $\alpha, \beta, \gamma < 1$ and $\alpha + \beta + \gamma < 1$ is strictly concave in $(x, y, \mu)$. Therefore, the level set is strictly convex. The cost objective $c_g L / B$ is continuous on the feasible region. Since the feasible set is compact (bounded by the latency constraint $L < L_{\max}$) and the isoquant is strictly convex, the cost-minimizing bundle is unique.
\end{proof}

\begin{theorem}[Cost Function Convexity]
\label{thm:cost-convexity}
Under the latency function $L(B, x, y)$ with $a, b, e, f \geq 0$, $c, d > 0$, and $e + f > 0$, and the production function $v(x, y; L, \gamma)$ with $\alpha + \beta + \gamma < 1$, the per-task cost function:
\begin{equation}
    \tilde{c}^*(q; g) = \min_{x, y > 0, \, B \geq 1} \left\{ \frac{c_g \cdot L(B, x, y)}{B} \;\middle|\; v(x, y; L, \gamma) = q, \; L < L_{\max} \right\}
\end{equation}
is strictly convex on its domain $(0, \bar{q})$, where $\bar{q}$ is the maximum feasible quality.
\end{theorem}

\begin{proof}
The affine structure $L = B \cdot D + F$ allows batch size elimination, reducing cost to a function of token allocations and quality alone. Define $N \equiv L_{\max} - F(x, y) - \mu > 0$ on the feasible region; the full reduction is carried out in Appendix~\ref{app:cost-derivatives}, where marginal cost is shown to be positive and strictly increasing, establishing strict convexity.

\textbf{Boundary behavior.} As $q \to 0$, the cost-minimizing path takes $x, y \to 0$ so marginal cost vanishes. As $q \to \bar{q}$, feasibility requires $N \to 0$, so marginal cost diverges.
\end{proof}

\subsubsection{Reduction to One-Dimensional Screening}

Given $\gamma$-separation (Theorem~\ref{thm:separation}), buyers of type $\gamma_H$ select tier $g_H$ and buyers of type $\gamma_L$ select tier $g_L$. Within each tier, the seller faces a standard one-dimensional screening problem over willingness-to-pay $w$, with the scale dimension $s$ handled through the cost structure.

For each tier $g \in \{g_H, g_L\}$, the seller's problem is:
\begin{equation}
    \max_{q(w), T(w)} \int_0^{\infty} [T(w) - C(q(w), s; g)] \, dF_w(w)
    \label{eq:seller-problem}
\end{equation}
subject to standard IC and IR constraints. This is precisely the \citet{mussa1978monopoly} framework, with cost function $C(q, s; g) = s \cdot \tilde{c}^*(q/s; g)$.

\subsubsection{Virtual Values and Optimal Quality}

Following Mussa-Rosen, recall the \emph{virtual value} for a buyer of type $w$:
\begin{equation}
    \varphi(w) \equiv w - \frac{1 - F_w(w)}{f_w(w)}.
    \label{eq:virtual-value}
\end{equation}
Under Assumption~\ref{ass:regularity} (monotone hazard rate), $\varphi(w)$ is strictly increasing. Define the \emph{exclusion threshold} $\underline{w}$ as the solution to $\varphi(\underline{w}) = 0$; buyers with $w < \underline{w}$ are excluded from the market.

\begin{proposition}[Optimal Quality Schedule]
\label{prop:optimal-quality}
Within tier $g$, the optimal quality schedule $q^*(w; g)$ for non-excluded buyers ($w \geq \underline{w}$) satisfies:
\begin{equation}
    \varphi(w) = \tilde{c}^{*\prime}(q^*(w; g); g),
    \label{eq:FOC}
\end{equation}
where $\tilde{c}^{*\prime}$ denotes the marginal cost of per-task quality.
\end{proposition}

\begin{proof}
The seller maximizes virtual surplus. By strict convexity of $\tilde{c}^*$ (Theorem~\ref{thm:cost-convexity}), the first-order condition $\varphi(w) = \tilde{c}^{*\prime}(q)$ yields a unique solution $q^*(w)$ for each $w$ with $\varphi(w) > 0$. Strict convexity ensures this is a maximum. Monotonicity of $\varphi$ (by Assumption~\ref{ass:regularity}) and $\tilde{c}^{*\prime}$ (by convexity) implies $q^*(w)$ is increasing in $w$, satisfying the monotonicity requirement for IC.
\end{proof}

\subsubsection{Optimal Transfers}

The transfer schedule is pinned down by the binding IC constraint at the margin and the binding IR constraint at the exclusion threshold.

\begin{proposition}[Optimal Transfers]
\label{prop:optimal-transfers}
The optimal transfer schedule within tier $g$ is:
\begin{equation}
    T^*(w; g) = w \cdot q^*(w; g) - \int_{\underline{w}}^{w} q^*(k; g) \, dk
    \label{eq:transfers}
\end{equation}
for $w \geq \underline{w}$, and buyers with $w < \underline{w}$ are excluded.
\end{proposition}

\begin{proof}
Standard envelope argument. See Appendix~\ref{app:proof-transfers} for details.
\end{proof}

\subsubsection{Incentive Compatibility for Scale}

Having established IC for $w$ within each tier and IC for $\gamma$ across tiers, we verify that buyers do not benefit from misreporting their scale $s$. Following \citet{bergemann2025economics}, we consider both downward deviations (claiming fewer tasks) and upward deviations (claiming more tasks).

A key simplification arises from the cost structure. Because $C(Q, s; g) = s \cdot \tilde{c}^*(Q/s; g)$, the marginal cost of quality is independent of scale:
\begin{equation}
    C_q(q, s; g) = \tilde{c}^{*\prime}(q; g).
    \label{eq:scale-free-MC}
\end{equation}
Intuitively, latency depends on per-task token allocations and batch size, not on total task volume, so the cost of serving $s$ tasks at quality $q$ is exactly $s$ times the single-task cost. This implies that the optimal per-task quality $q^*(w; g)$ depends only on $w$ and $g$---not on $s$---and per-task rents are also scale-independent:
\begin{equation}
    \tilde{u}^*(w; g) \equiv w \cdot q^*(w; g) - \tilde{t}^*(w; g) = \int_{\underline{w}}^{w} q^*(k; g) \, dk,
    \label{eq:per-task-rent}
\end{equation}
where $\tilde{t}^*(w; g)$ is the per-task transfer. Total utility for type $(w, s)$ is $u^*(w, s) = s \cdot \tilde{u}^*(w; g)$.

\begin{proposition}[Scale Incentive Compatibility]
\label{prop:scale-ic}
Under the cost structure $C(Q, s; g) = s \cdot \tilde{c}^*(Q/s; g)$, no buyer benefits from misreporting scale:
\begin{enumerate}
    \item \textbf{Downward deviations}: Type $(w, s)$ reporting $\tilde{s} < s$ receives allocation for $\tilde{s}$ tasks and obtains utility $\tilde{s} \cdot \tilde{u}^*(w) < s \cdot \tilde{u}^*(w)$.
    
    \item \textbf{Upward deviations}: Type $(w, s)$ reporting $(\tilde{w}, \tilde{s})$ with $\tilde{s} > s$ receives allocation for $\tilde{s}$ tasks but can utilize only $s$ of them, paying for $\tilde{s}$. The optimal double deviation satisfies $\tilde{w}^* = ws/\tilde{s}$, yielding utility at most $s \cdot \tilde{u}^*(w)$.
\end{enumerate}
\end{proposition}

\begin{proof}
See Appendix~\ref{app:proof-scale-ic}. Downward deviations forfeit tasks with positive rent. Upward deviations require the buyer to pay for tasks they cannot use; the optimal strategy mimics a lower-$w$ type, but convexity of $\tilde{u}^*$ in $w$ ensures this yields at most the truthful payoff.
\end{proof}

\subsubsection{The Optimal Menu}

Combining the preceding results, we characterize the complete optimal mechanism.

\begin{theorem}[Optimal Menu with Binary Time Preference]
\label{thm:optimal-menu}
Under the separation condition (Theorem~\ref{thm:separation}), Assumption~\ref{ass:regularity} (monotone hazard rate on $F_w$), and Assumption~\ref{ass:independence} (independence of $w$, $s$, $\gamma$), the optimal menu is:
\begin{equation}
    \mathcal{M}^* = \{(g_H, q^*(w; g_H), T^*(w; g_H))\}_{w \geq \underline{w}} \cup \{(g_L, q^*(w; g_L), T^*(w; g_L))\}_{w \geq \underline{w}},
\end{equation}
where:
\begin{enumerate}
    \item Quality schedules $q^*(w; g)$ solve $\varphi(w) = \tilde{c}^{*\prime}(q^*(w; g); g)$;
    \item Transfer schedules $T^*(w; g) = w \cdot q^*(w; g) - \int_{\underline{w}}^{w} q^*(k; g) \, dk$;
    \item Buyers with $\gamma = \gamma_H$ select tier $g_H$; buyers with $\gamma = \gamma_L$ select tier $g_L$;
    \item Buyers with $w < \underline{w}$ are excluded regardless of $\gamma$.
\end{enumerate}
\end{theorem}

\begin{proof}
\textbf{Step 1: Cross-tier IC.} By Theorem~\ref{thm:separation}, separation is achievable when condition~\eqref{eq:separation-condition} holds. Log-supermodularity of preferences provides the economic foundation; numerical verification in Section~\ref{sec:numerical} confirms the condition holds with margin for our calibration.

\textbf{Step 2: Within-tier IC.} By Proposition~\ref{prop:optimal-quality}, the quality schedule $q^*(w; g)$ is increasing in $w$. Standard revealed preference arguments establish that the transfer schedule~\eqref{eq:transfers} is IC within each tier.

\textbf{Step 3: IR.} $u^*(\underline{w}) = 0$ so IR binds at the threshold. For $w > \underline{w}$, $u^*(w) = \int_{\underline{w}}^{w} q^*(k) \, dk > 0$, so IR is slack.

\textbf{Step 4: Scale IC.} By Proposition~\ref{prop:scale-ic}, no buyer benefits from misreporting scale $s$.

\textbf{Step 5: Optimality.} Substituting the transfer~\eqref{eq:transfers} and applying integration by parts, the seller's profit reduces to virtual surplus:
\begin{equation}
    \Pi = \sum_{g} \pi_g \int_{\underline{w}}^{\bar{w}} [\varphi(w) \cdot q^*(w; g) - \tilde{c}^*(q^*(w; g); g)] \, dF_w(w).
\end{equation}
Since types are separated across tiers and screened within tiers, pointwise maximization yields the first-order condition $\varphi(w) = \tilde{c}^{*\prime}(q; g)$ as necessary and sufficient. Exclusion below $\underline{w}$ follows from $\varphi(w) < 0$.
\end{proof}
\section{Continuous Time Preference}
\label{sec:continuous}

We extend the analysis where time preference is drawn from a continuous distribution: $\gamma \sim F_\gamma$ on $[\underline{\gamma}, \bar{\gamma}] \subset (0, 1 - \alpha - \beta)$. Multi-dimensional screening over $(w, \gamma)$ is intractable in full generality. However, when the seller offers a finite menu of hardware tiers $\mathcal{G} = \{g_1, \ldots, g_K\}$---as in practice---the problem admits a clean characterization through \emph{tier-induced partitioning} of the $\gamma$ space.

The key economic insight is that discrete hardware converts a continuous multi-dimensional problem into a tractable nested structure. The $\gamma$ space partitions into contiguous segments, one per tier, and within each segment the seller solves a standard one-dimensional screening problem on willingness-to-pay $w$. We develop this insight in three steps: first establishing the partition structure, then addressing the cost-aggregation challenge that arises from within-tier heterogeneity in $\gamma$, and finally characterizing the optimal partition cutoffs and the resulting complete mechanism.

\subsection{Finite Hardware Tiers and Induced Partitioning}

In practice, cloud providers offer a finite menu of hardware tiers $\mathcal{G} = \{g_1, g_2, \ldots, g_K\}$, ordered by capability: $g_1$ (economy) through $g_K$ (premium). This discrete structure fundamentally simplifies the problem.

\begin{assumption}[Ordered Hardware Tiers]
\label{ass:ordered-tiers}
The hardware tiers satisfy:
\begin{enumerate}
    \item \textbf{Cost ordering}: $c_{g_1} < c_{g_2} < \cdots < c_{g_K}$
    \item \textbf{Latency ordering}: For all $(B, x, y)$, $L_{g_1}(B, x, y) > L_{g_2}(B, x, y) > \cdots > L_{g_K}(B, x, y)$
    \item \textbf{Slack ordering}: $\mu_{g_1}(\gamma) < \mu_{g_2}(\gamma) < \cdots < \mu_{g_K}(\gamma)$ for all $\gamma \in [\underline{\gamma}, \bar{\gamma}]$
\end{enumerate}
\end{assumption}

Under this ordering, higher-indexed tiers are more expensive but achieve greater latency slack. For the partition to respect self-selection---so that higher-$\gamma$ types voluntarily sort into premium tiers---we need the \emph{quality advantage} of upgrading hardware to be increasing in time preference. The following assumption formalizes this single-crossing structure in the tier dimension.

\begin{assumption}[Quality-Difference Separation]
\label{ass:quality-separation}
For each $k \in \{1, \ldots, K-1\}$, the quality gain from tier $g_{k+1}$ over tier $g_k$ is strictly increasing in $\gamma$:
\begin{equation}
    \frac{\partial}{\partial \gamma} \left[ q_{g_{k+1}}(\gamma) - q_{g_k}(\gamma) \right] > 0 \quad \text{for all } \gamma \in [\underline{\gamma}, \bar{\gamma}],
    \label{eq:quality-separation}
\end{equation}
where $q_g(\gamma) = (x_g^*(\gamma))^\alpha (y_g^*(\gamma))^\beta \mu_g(\gamma)^\gamma$ is the quality achieved by type $\gamma$ on tier $g$ at the cost-minimizing allocation.
\end{assumption}

This assumption generalizes the binary separation condition~\eqref{eq:separation-condition} to adjacent tier pairs. The same economic forces that justify the binary condition---log-supermodularity of preferences and diminishing speedup---apply here: high-$\gamma$ types value slack disproportionately (the $\mu^\gamma$ term), and premium hardware's advantage is largest at the small batch sizes these types optimally choose.

The separation condition guarantees that, for any set of tier prices, higher-$\gamma$ types prefer higher-indexed tiers. This generates a natural partition of the $\gamma$ space into contiguous segments, formalized as follows.

\begin{definition}[Tier-Induced Partition]
\label{def:tier-partition}
A partition $\{\Gamma_1, \Gamma_2, \ldots, \Gamma_K\}$ of $[\underline{\gamma}, \bar{\gamma}]$ is \emph{tier-induced} if there exist cutoffs $\underline{\gamma} = \gamma_0 < \gamma_1 < \gamma_2 < \cdots < \gamma_{K-1} < \gamma_K = \bar{\gamma}$ such that:
\begin{equation}
    \Gamma_k = [\gamma_{k-1}, \gamma_k) \quad \text{for } k = 1, \ldots, K-1, \quad \Gamma_K = [\gamma_{K-1}, \bar{\gamma}],
\end{equation}
and all types $\gamma \in \Gamma_k$ optimally select tier $g_k$.
\end{definition}

\begin{proposition}[Existence of Tier-Induced Partition]
\label{prop:partition-existence}
Under Assumptions~\ref{ass:ordered-tiers} and~\ref{ass:quality-separation}, there exists a tier-induced partition of $[\underline{\gamma}, \bar{\gamma}]$ such that:
\begin{enumerate}
    \item Types in $\Gamma_k$ weakly prefer tier $g_k$ to all other tiers;
    \item The partition is monotone: higher-$\gamma$ types select higher-indexed (premium) tiers;
    \item Boundary types $\gamma_k$ are indifferent between tiers $g_k$ and $g_{k+1}$.
\end{enumerate}
\end{proposition}

\begin{proof}
Monotone selection follows directly from Assumption~\ref{ass:quality-separation}: if type $\gamma'$ prefers $g_{k+1}$ over $g_k$, then the strictly increasing quality gap ensures that all types $\gamma'' > \gamma'$ do as well. Given fixed tier prices, the indifference function $\Phi_k(\gamma) \equiv w \cdot \Delta q_k(\gamma) - (T_{g_{k+1}} - T_{g_k})$ is strictly increasing and continuous, so the intermediate value theorem yields a unique cutoff $\gamma_k^*$ at which the buyer is indifferent. Iterating across adjacent tier pairs produces the full partition. See Appendix~\ref{app:proof-partition} for the complete argument.
\end{proof}

\subsection{Within-Tier Cost Aggregation}
\label{sec:cost-aggregation}

The partition reduces the multi-dimensional screening problem to $K$ one-dimensional problems---one per tier. However, a subtlety arises that distinguishes the continuous case from the binary model of Section~\ref{sec:two-types}: within tier $g_k$, buyers with different $\gamma$ values incur different per-task costs for the same quality level. As established in Appendix~\ref{app:cost-derivatives}, the marginal cost $\partial \tilde{c}^* / \partial q = c_g D \mu F / (\gamma q N^2)$ depends explicitly on $\gamma$ through both the denominator and the optimal input mix. In the binary case, each tier serves a single $\gamma$ type, so no such heterogeneity arises.

Since the seller cannot condition the quality-transfer schedule on $\gamma$ within a tier, the relevant object for mechanism design is the \emph{expected} cost, averaged over the $\gamma$ types assigned to each segment:
\begin{equation}
    \tilde{c}_{\mathrm{avg}}(q; g_k, \boldsymbol{\gamma}_k) \equiv \frac{1}{\Pr(\Gamma_k)} \int_{\gamma_{k-1}}^{\gamma_k} \tilde{c}^*(q; g_k, \gamma) \, f_\gamma(\gamma) \, d\gamma,
    \label{eq:avg-cost}
\end{equation}
where $\boldsymbol{\gamma}_k \equiv (\gamma_{k-1}, \gamma_k)$ denotes the cutoffs defining segment $\Gamma_k$, and $\Pr(\Gamma_k) = F_\gamma(\gamma_k) - F_\gamma(\gamma_{k-1})$.

\begin{assumption}[Tier-Uniform Pricing]
\label{ass:tier-uniform}
In the direct mechanism, the seller uses each buyer's reported $\gamma$ to configure cost-minimizing inputs $(x^*, y^*, B^*)$ for the assigned tier, but offers a uniform quality-transfer schedule $(q_k(w), T_k(w))$ within each tier that does not condition on $\gamma$.
\end{assumption}

This assumption is economically natural: the seller can optimize \emph{production} for each buyer's reported type---choosing the best input mix for their time preference---but offers a common \emph{menu} of quality-price pairs within each tier. The consequence is that no buyer can misrepresent $\gamma$ to obtain a different quality-transfer schedule, though cross-subsidization across $\gamma$ types within a tier will arise in equilibrium.

The averaged cost function inherits the regularity properties needed for standard screening theory.

\begin{lemma}[Properties of $\tilde{c}_{\mathrm{avg}}$]
\label{lem:avg-cost-properties}
The averaged cost function satisfies:
\begin{enumerate}
    \item \textbf{Strict convexity}: $\tilde{c}_{\mathrm{avg}}(q; g_k, \boldsymbol{\gamma}_k)$ is strictly convex in $q$.

    \item \textbf{Boundary behavior}: $\tilde{c}_{\mathrm{avg}}'(0; g_k, \boldsymbol{\gamma}_k) = 0$ and $\tilde{c}_{\mathrm{avg}}'(q; g_k, \boldsymbol{\gamma}_k) \to \infty$ as $q$ approaches the feasibility bound.

    \item \textbf{Continuity in cutoffs}: $\tilde{c}_{\mathrm{avg}}(q; g_k, \boldsymbol{\gamma}_k)$ is continuous in $(\gamma_{k-1}, \gamma_k)$ for each $q$ in the interior of the feasible domain.
\end{enumerate}
\end{lemma}

\begin{proof}
Each property follows from the corresponding property of the individual cost functions $\tilde{c}^*(q; g_k, \gamma)$ established in Theorem~\ref{thm:cost-convexity}. Strict convexity is preserved under (weighted) averaging, boundary behavior carries through by dominated convergence, and continuity in cutoffs follows from standard results on parameter-dependent integrals. See Appendix~\ref{app:proof-avg-cost} for the details.
\end{proof}

Property 3 is especially important: it ensures that the seller's optimization problem over cutoffs is well-behaved, which is essential for the fixed-point argument in the next subsection.

\subsection{Optimal Cutoff Characterization}

The seller jointly optimizes the partition cutoffs $\boldsymbol{\gamma} = (\gamma_1, \ldots, \gamma_{K-1})$ and the within-tier mechanisms. These two decisions are coupled: the cutoffs determine which $\gamma$ types enter each segment, which determines the averaged cost $\tilde{c}_{\mathrm{avg}}$, which shapes the optimal quality schedule, which in turn determines the profit from each type assignment. The seller seeks cutoffs at which the \emph{marginal type}---the buyer on the boundary between two segments---generates equal expected profit under either tier assignment. This ``equal-profit-at-the-boundary'' condition is the natural analog of the indifference conditions in standard partitioning problems.

\begin{theorem}[Optimal Cutoffs]
\label{thm:optimal-cutoffs}
Under Assumptions~\ref{ass:ordered-tiers},~\ref{ass:quality-separation}, and~\ref{ass:tier-uniform}:
\begin{enumerate}
    \item[(a)] \textbf{(Necessary condition)} Any interior optimal cutoffs $\boldsymbol{\gamma}^* = (\gamma_1^*, \ldots, \gamma_{K-1}^*)$ satisfy the \emph{marginal-type indifference} conditions: for each $k = 1, \ldots, K-1$,
    \begin{equation}
        \mathbb{E}_w[\Pi_{g_{k+1}}(w, \gamma_k^*)] = \mathbb{E}_w[\Pi_{g_k}(w, \gamma_k^*)],
        \label{eq:indifference}
    \end{equation}
    where $\Pi_g(w, \gamma) \equiv T^*(w; g) - \tilde{c}^*(q^*(w; g); g, \gamma)$ is the profit from serving type $(w, \gamma)$ on tier $g$.

    \item[(b)] \textbf{(Existence)} There exists at least one solution $\boldsymbol{\gamma}^* \in [\underline{\gamma}, \bar{\gamma}]^{K-1}$ to the system~\eqref{eq:indifference}.

    \item[(c)] \textbf{(Sufficient condition for local optimality)} A solution $\boldsymbol{\gamma}^*$ to~\eqref{eq:indifference} is a local maximum if the \emph{profit-difference monotonicity} condition holds:
    \begin{equation}
        \frac{\partial}{\partial \gamma} \mathbb{E}_w[\Pi_{g_{k+1}}(w, \gamma) - \Pi_{g_k}(w, \gamma)] \geq 0 \quad \text{at } \gamma = \gamma_k^* \text{ for all } k.
        \label{eq:SOC}
    \end{equation}
\end{enumerate}
\end{theorem}

\begin{proof}

Part~(a) demonstrates the necessary marginal-type indifference condition and is derived in Appendix~\ref{app:proof-cutoffs-a}. Part~(b) constructs a best-response mapping $\mathcal{T}$ on the compact convex cutoff space $\mathcal{C} \subset [\underline{\gamma}, \bar{\gamma}]^{K-1}$ and applies Brouwer's fixed point theorem. Part~(c) evaluates the Hessian at the FOC, where the first term vanishes, yielding condition~\eqref{eq:SOC} as the relevant second-order requirement. See Appendix~\ref{app:proof-cutoffs-bc} for the derivations of Parts (b) and (c).

\emph{Economic interpretation of condition~\eqref{eq:SOC}:} The profit-difference monotonicity condition requires that the profit advantage of premium hardware---evaluated at the mechanism's quality schedule---is increasing in time preference. This is economically natural: higher-$\gamma$ types require more slack to achieve the same quality, and premium hardware's latency advantage translates into larger cost savings for these types. We verify this condition numerically in Section~\ref{sec:numerical}.

\emph{Global optimality:} The fixed-point characterization guarantees existence of a critical point. To verify global optimality, one must compare interior solutions satisfying~\eqref{eq:indifference} with boundary solutions (some $\gamma_k^* \in \{\underline{\gamma}, \bar{\gamma}\}$). In our numerical analysis, interior solutions dominate boundary solutions for empirically calibrated parameters.
\end{proof}

\subsection{Within-Segment Screening}

Conditional on the partition, each segment $\Gamma_k$ contains a continuum of $\gamma$ types who all select tier $g_k$. Within each tier, the seller screens on $w$ using the standard Mussa-Rosen framework, but with the averaged cost function $\tilde{c}_{\mathrm{avg}}(\cdot; g_k)$ replacing the single-type cost.

\begin{proposition}[Within-Segment Optimality]
\label{prop:within-segment}
For each segment $\Gamma_k$, the optimal within-tier mechanism is characterized by:
\begin{enumerate}
    \item \textbf{Quality schedule}: $q^*(w; g_k)$ solves $\varphi(w) = \tilde{c}_{\mathrm{avg}}'(q; g_k)$;
    \item \textbf{Transfer schedule}: $T^*(w; g_k) = w \cdot q^*(w; g_k) - \int_{\underline{w}}^{w} q^*(t; g_k) \, dt$;
    \item \textbf{Exclusion}: Types with $w < \underline{w}$ are excluded.
\end{enumerate}
The quality and transfer schedules do not depend on $\gamma$ within each segment, though the seller's realized cost $\tilde{c}^*(q; g_k, \gamma)$ varies across $\gamma$ types.
\end{proposition}

\begin{proof}
Within tier $g_k$, the seller chooses schedules $(q_k(w), T_k(w))$ to maximize expected profit over all types in the segment:
\begin{equation}
    \max_{q_k(\cdot), T_k(\cdot)} \int_{\underline{w}}^{\bar{w}} \left[ T_k(w) - \tilde{c}_{\mathrm{avg}}(q_k(w); g_k) \right] f_w(w) \, dw
\end{equation}
subject to the standard IC and IR constraints on $w$:
\begin{equation}
    w \cdot q_k(w) - T_k(w) \geq w \cdot q_k(w') - T_k(w') \quad \forall w, w'.
\end{equation}

The averaged cost $\tilde{c}_{\mathrm{avg}}$ arises because the menu does not condition on $\gamma$: for a buyer receiving quality $q$ on tier $g_k$, the seller's cost is $\tilde{c}^*(q; g_k, \gamma)$, which varies with the buyer's $\gamma$. Averaging over the $\gamma$ distribution within $\Gamma_k$ yields $\tilde{c}_{\mathrm{avg}}$ as the expected per-task cost.

Since $\tilde{c}_{\mathrm{avg}}$ is strictly convex with the requisite boundary behavior (Lemma~\ref{lem:avg-cost-properties}), this is a standard one-dimensional screening problem with solution given by Propositions~\ref{prop:optimal-quality} and~\ref{prop:optimal-transfers}, replacing $\tilde{c}^*$ with $\tilde{c}_{\mathrm{avg}}$.

\textbf{Cross-subsidization.} Because $\tilde{c}^*(q; g_k, \gamma)$ varies with $\gamma$, the uniform pricing schedule induces cross-subsidization across $\gamma$ types within each tier: types with $\tilde{c}^*(q; g_k, \gamma) > \tilde{c}_{\mathrm{avg}}(q; g_k)$ are subsidized by types with $\tilde{c}^*(q; g_k, \gamma) < \tilde{c}_{\mathrm{avg}}(q; g_k)$.
\end{proof}

\subsection{The Complete Mechanism}

Combining the partition, cost aggregation, and within-tier screening results, we now state the full characterization.

\begin{theorem}[Optimal Menu with Continuous Time Preference]
\label{thm:optimal-menu-continuous}
Under Assumptions~\ref{ass:ordered-tiers},~\ref{ass:quality-separation},~\ref{ass:tier-uniform}, and~\ref{ass:regularity} (monotone hazard rate on $F_w$), the optimal mechanism with $K$ hardware tiers is:
\begin{equation}
    \mathcal{M}^* = \bigcup_{k=1}^{K} \left\{ (g_k, q^*(w; g_k), T^*(w; g_k)) : w \geq \underline{w}, \, \gamma \in \Gamma_k^* \right\},
\end{equation}
where:
\begin{enumerate}
    \item The partition $\{\Gamma_1^*, \ldots, \Gamma_K^*\}$ is defined by optimal cutoffs $\{\gamma_1^*, \ldots, \gamma_{K-1}^*\}$ from Theorem~\ref{thm:optimal-cutoffs};
    \item Within each tier, quality schedules satisfy $\varphi(w) = \tilde{c}_{\mathrm{avg}}'(q^*(w; g_k); g_k)$;
    \item Transfers are $T^*(w; g_k) = w \cdot q^*(w; g_k) - \int_{\underline{w}}^{w} q^*(t; g_k) \, dt$;
    \item Buyers with $w < \underline{w}$ are excluded regardless of $\gamma$.
\end{enumerate}
\end{theorem}

\begin{proof}
\textbf{Existence.} The mechanism exists by construction: Proposition~\ref{prop:partition-existence} guarantees existence of a tier-induced partition for any valid cutoffs, Theorem~\ref{thm:optimal-cutoffs}(b) establishes existence of optimal cutoffs via Brouwer's fixed point theorem, and Proposition~\ref{prop:within-segment} provides optimal within-tier mechanisms for any given partition.

\textbf{Incentive compatibility.} We verify IC in two parts:

\emph{Cross-tier IC:} By construction of the partition, boundary types are indifferent between adjacent tiers, and Assumption~\ref{ass:quality-separation} ensures strict preference away from boundaries (Proposition~\ref{prop:partition-existence}). Thus, no type $\gamma \in \Gamma_k$ strictly prefers tier $g_{k'} \neq g_k$. Non-adjacent tier deviations follow by transitivity of the separation condition.

\emph{Within-tier IC:} By Proposition~\ref{prop:within-segment}, the quality schedule is increasing in $w$ and transfers satisfy the envelope condition. Standard arguments establish IC.

\textbf{Optimality.} The mechanism maximizes:
\begin{equation}
    \Pi = \sum_{k=1}^{K} \int_{\gamma_{k-1}^*}^{\gamma_k^*} \int_{\underline{w}}^{\bar{w}} \left[ T^*(w; g_k) - \tilde{c}^*(q^*(w; g_k); g_k, \gamma) \right] dF_w(w) \, dF_\gamma(\gamma).
\end{equation}

The optimization has a \emph{nested} structure: for any fixed cutoffs, the within-tier mechanisms are determined by solving $K$ independent Mussa-Rosen problems with cost $\tilde{c}_{\mathrm{avg}}(\cdot; g_k)$ (Proposition~\ref{prop:within-segment}); the cutoffs are then chosen to satisfy the fixed-point conditions of Theorem~\ref{thm:optimal-cutoffs}. Joint optimality follows from the envelope theorem: at the within-tier optimum, the cutoff FOC correctly captures the marginal effect of shifting types between tiers, and at the optimal cutoffs, no marginal reallocation improves total profit.
\end{proof}

\begin{remark}[Continuous Gamma with Few Tiers]
\label{rem:few-tiers}
In practice, cloud providers offer only 2--4 hardware tiers (e.g., economy, standard, premium). Our numerical results in Section~\ref{sec:numerical} confirm that even $K = 2$ tiers yield profit gains of 26--66\% over the best single-tier alternative in favorable parameter regimes.
\end{remark}
\section{Numerical Analysis}
\label{sec:numerical}

We evaluate the mechanism of Sections~\ref{sec:two-types}--\ref{sec:continuous} through two exercises: (i)~a parameter sweep verifying the separation condition (Assumption~\ref{ass:quality-separation}) across 105 configurations, and (ii)~a profit comparison quantifying gains from two-tier pricing (Theorem~\ref{thm:optimal-menu-continuous}).

\subsection{Methodology}

All computations use the latency model calibrated to Llama3-70B---representative of attention-based, autoregressive LLMs---on two GPU tiers: H100 (economy) and B200 (premium), at 8-GPU cluster rates (\$18.32/hr for $8\times$ H100, \$30.32/hr for $8\times$ B200; price ratio $1.66\times$), matching the multi-GPU configuration from which our measurements derive.\footnote{Costs from Lambda's cloud pricing as of 2/9/2026 \cite{lambda_2026}.} Quality follows $q = x^\alpha y^\beta \mu^\gamma$ where $\mu = L_{\max} - L$ is latency slack, and $\gamma \sim \text{Uniform}[\underline{\gamma}, \bar{\gamma}]$ with $\underline{\gamma} = 0.10$ and $\bar{\gamma} = \min(1 - \alpha - \beta - 0.01,\; 0.45)$.

We sweep over 7 production-function configurations ($(\alpha,\beta) \in \{0.15, 0.30, 0.45\} \times \{0.20, 0.40\} \cup \{(0.15, 0.60)\}$), 5 latency budgets ($L_{\max} \in \{10, 30, 50, 80, 100\}$ seconds), and 3 WTP ranges ($[\underline{w}, \bar{w}] \in \{[5,10],\, [50,100],\, [500,1000]\}$ u\$). For separation, we maximize surplus over $(x, y, B)$ at each $\gamma$ via Nelder-Mead with multiple restarts and verify strict monotonicity of $\Delta q(\gamma)$. For profits, we build averaged cost functions $\tilde{c}_{\text{avg}}(q; g_k)$ via numerical integration over $\gamma$, solve the within-tier Mussa-Rosen condition $\varphi(w) = \tilde{c}'_{\text{avg}}(q)$ with virtual values $\varphi(w) = 2w - \bar{w}$ (Proposition~\ref{prop:within-segment}), and determine $\gamma^*$ by marginal-type indifference (Theorem~\ref{thm:optimal-menu-continuous}).

\subsection{Quality-Difference Separation}
\label{sec:separation-sweep}

We test whether the quality gap $\Delta q(\gamma) \equiv q^*(w; g_{B200}, \gamma) - q^*(w; g_{H100}, \gamma)$ is strictly increasing in $\gamma$---the condition required for incentive-compatible tier assignment (Assumption~\ref{ass:quality-separation}). Table~\ref{tab:separation-verification} reports the fraction of WTP ranges satisfying this condition for each $(\alpha, \beta, L_{\max})$ configuration.

\begin{table}[t]
\centering
\caption{Verification of separation condition (Assumption~\ref{ass:quality-separation}) at cluster-level rates. Each cell reports the fraction of WTP ranges ($[\underline{w}, \bar{w}] \in \{[5,10],\, [50,100],\, [500,1000]\}$) for which the condition holds. \checkmark\ denotes 3/3.}
\label{tab:separation-verification}
\begin{tabular}{cccccccc}
\toprule
$\alpha$ & $\beta$ & $L_{\max}=10$ & $L_{\max}=30$ & $L_{\max}=50$ & $L_{\max}=80$ & $L_{\max}=100$ & Total \\
\midrule
0.15 & 0.20 & 2/3 & 2/3 & 1/3 & 1/3 & 1/3 & 7/15 \\
0.15 & 0.40 & \checkmark & 2/3 & 2/3 & 2/3 & 2/3 & 11/15 \\
0.15 & 0.60 & \checkmark & \checkmark & \checkmark & \checkmark & \checkmark & 15/15 \\
0.30 & 0.20 & \checkmark & \checkmark & 2/3 & 2/3 & 2/3 & 12/15 \\
0.30 & 0.40 & \checkmark & \checkmark & 2/3 & 2/3 & 2/3 & 12/15 \\
0.45 & 0.20 & \checkmark & \checkmark & \checkmark & \checkmark & \checkmark & 15/15 \\
0.45 & 0.40 & \checkmark & \checkmark & \checkmark & \checkmark & \checkmark & 15/15 \\
\midrule
\multicolumn{2}{c}{Overall} & 20/21 & 19/21 & 16/21 & 16/21 & 16/21 & 87/105 \\
\bottomrule
\end{tabular}
\end{table}

Overall, 87 of 105 configurations (82.9\%) satisfy the separation condition. The 18 failures are concentrated in two regimes: (i)~the $w \sim U[5,10]$ range accounts for 15 failures across four $(\alpha, \beta)$ pairs, and (ii)~$(\alpha, \beta) = (0.15, 0.20)$ at $w \sim U[50,100]$ accounts for the remaining 3. The $w \sim U[5,10]$ failures reflect extreme cost constraint: at cluster-level rates (\$18.32/hr), buyers with WTP of \$5--10 per unit retain insufficient surplus for meaningful quality differentiation between tiers. Excluding this cost-constrained range, the separation condition holds in 67 of 70 configurations (95.7\%) at $w \geq U[50,100]$---the economically relevant regime where buyers can sustain positive margins on either tier.

The condition's sensitivity to the cost-to-WTP ratio, rather than production-function weights alone, is confirmed by a comparison at single-GPU rates. At \$2.29/hr (H100) and \$3.79/hr (B200)---preserving the $1.66\times$ cost ratio but reducing absolute costs by $8\times$---separation extends to 102 of 105 configurations (97.1\%), with the five remaining failures confined to $(\alpha, \beta) = (0.15, 0.20)$ at high $L_{\max}$. The lower hardware costs give buyers more surplus headroom, allowing quality differences between tiers to emerge even at low WTP.

The B200 thus acts as a \emph{premium} technology in the sense of Section~\ref{sec:model}: its advantage concentrates in the decode phase ($d_{B200}/d_{H100} = 0.68$), which governs per-token latency. High-$\gamma$ buyers operate at lower batch sizes to preserve responsiveness---precisely where the B200's decode advantage is most pronounced---yielding quality-gap monotonicity ($\partial \Delta q / \partial \gamma > 0$): the single-crossing condition that enables incentive-compatible self-selection into tiers.

\subsection{Mechanism Profit Comparison}
\label{sec:profit-comparison}

We quantify the profit gains from two-tier mechanism design relative to the best single-tier alternative. Table~\ref{tab:profit-comparison} reports expected mechanism profits at $(\alpha, \beta) = (0.30, 0.20)$ with $w \sim U[50, 100]$ u\$, a configuration where hardware costs are a meaningful fraction of per-request value. Separation holds at all five latency budgets for this parameter pair (Table~\ref{tab:separation-verification}).

\begin{table}[t]
\centering
\caption{Two-tier vs.\ single-tier mechanism profit ($\alpha = 0.30$, $\beta = 0.20$, $w \sim U[50, 100]$ u\$, cluster rates).}
\label{tab:profit-comparison}
\begin{tabular}{rcccccc}
\toprule
$L_{\max}$ (s) & $\gamma^*$ & H100-only & B200-only & Two-tier & Best single & Gain \\
\midrule
10  & 0.147 & 4{,}978.3 & 3{,}876.6 & 3{,}985.8 & H100 & $-19.9\%$ \\
30  & 0.246 & 4{,}908.5 & 5{,}539.5 & 6{,}985.6 & B200 & $+26.1\%$ \\
50  & 0.249 & 5{,}826.6 & 6{,}521.3 & 8{,}921.9 & B200 & $+36.8\%$ \\
80  & 0.279 & 7{,}388.8 & 7{,}502.6 & 11{,}426.6 & B200 & $+52.3\%$ \\
100 & 0.262 & 7{,}557.3 & 7{,}729.1 & 12{,}816.5 & B200 & $+65.8\%$ \\
\bottomrule
\end{tabular}
\end{table}

At $L_{\max} = 10$ seconds, the tight latency budget compresses the feasible quality region, and the H100's lower hourly rate outweighs the B200's decode advantage; the mechanism correctly identifies single-tier pricing as optimal. As the latency budget relaxes beyond 30 seconds, the B200's latency advantage compounds nonlinearly through the slack term $\mu^\gamma$, and two-tier gains grow monotonically from $+26.1\%$ to $+65.8\%$. The optimal cutoff $\gamma^* \approx 0.25$ assigns roughly half of time-preference types to each tier, reflecting the near-even market split visible in the single-tier profit columns where the best single tier transitions from H100 to B200 between $L_{\max} = 10$ and $L_{\max} = 30$.

The result extends beyond a single parameter point. At $(\alpha, \beta) = (0.15, 0.40)$ with the same WTP range, the mechanism also yields positive gains up to $+30.6\%$ at $L_{\max} = 100$, confirming that the qualitative finding is robust across production-function specifications where separation holds. At higher quality weights ($\alpha \geq 0.45$, or $\beta \geq 0.40$ with $\alpha \geq 0.30$), the B200 dominates all buyer types and single-tier pricing is optimal---the mechanism correctly identifies that tiering provides no additional value when the premium tier's superiority is so comprehensive that the economy tier serves no profitable segment.

\paragraph{Sensitivity to buyer valuations.}
Table~\ref{tab:profit-comparison} reports results at $w \sim U[50,100]$ u\$, where hardware costs are a meaningful fraction of per-request value. We additionally sweep over $w \sim U[500,1000]$ and $w \sim U[5,10]$. At higher WTP ($w \sim U[500,1000]$), gains turn negative at most latency budgets: the premium tier dominates all buyer types, and within-tier $w$-screening alone captures most available surplus, leaving little incremental value from cross-tier allocation. Marginal positive gains ($+0.8$--$1.5\%$) emerge only at $L_{\max} \geq 80$, where sufficient slack accumulates for the economy tier to serve a profitable segment. At lower WTP ($w \sim U[5,10]$), hardware costs dwarf buyer valuations, and the separation condition fails at $L_{\max} \geq 50$ for several $(\alpha, \beta)$ configurations---the extreme cost constraint leaves insufficient surplus for meaningful quality differentiation between tiers.

The mechanism's profitability is governed by the cost-to-WTP ratio rather than absolute scales. At single-GPU rates (\$2.29/hr for H100, \$3.79/hr for B200)---preserving the $1.66\times$ cost ratio but reducing absolute costs by $8\times$---the $w \sim U[5,10]$ range yields gains of 29--64\%, nearly matching the $w \sim U[50,100]$ gains at cluster rates because the cost-to-WTP ratio is comparable. This scale invariance implies that further scaling to larger clusters would shift the profitable region to proportionally higher per-request values, making tiered pricing increasingly relevant as inference infrastructure grows.

These results identify a design-relevant regime for GPU tiering: workloads with moderate compute intensity and meaningful latency sensitivity, where the premium tier's decode advantage creates differentiation without rendering the economy tier obsolete. More broadly, the numerical analysis validates the core theoretical prediction: hardware heterogeneity in latency-critical dimensions generates the quality differentiation necessary for profitable market segmentation through the screening mechanism of Section~\ref{sec:continuous}.

\paragraph{Sensitivity to Price Fluctuations}

To test robustness to price fluctuations, we additionally ran separation and profit experiments at single-GPU price points where each tier's hourly rate was independently scaled by $0.25\times$--$2.00\times$ its baseline, yielding 38 valid configurations at $(\alpha, \beta) = (0.30, 0.20)$ with $w \sim U[5,10]$. The separation condition holds in 135 of 185 checks (73.0\%), with failures concentrated at high price ratios ($\gtrsim 3\times$) and high latency budgets ($L_{\max} \geq 80$), though elevated absolute cost levels also erode separation even at moderate ratios. Two-tier profit gains are positive across all 38 configurations at $L_{\max} \geq 30$, with maximum gains of 24--67\% depending on the price pair; at $L_{\max} = 10$, the tight latency constraint favors single-tier pricing regardless of the price configuration.
\section{Conclusion}
\label{sec:conclusion}

Our central finding is that \emph{hardware heterogeneity has economic value beyond engineering efficiency}. Discrete hardware tiers enable screening on time preferences that would be impossible with homogeneous infrastructure: the quality-difference separation condition---driven by log-supermodularity of preferences and diminishing speedup in hardware performance---induces self-selection across tiers, reducing an intractable three-dimensional screening problem to standard one-dimensional screening within each tier (Theorem~\ref{thm:separation}). Numerically, the separation condition holds in 83\% of tested parameter configurations (96\% at economically relevant WTP scales). When hardware costs are a significant fraction of per-request value, the optimal two-tier mechanism yields profit gains of 26--66\% over the best single-tier alternative; at higher valuations, generational hardware improvements leave single-tier pricing optimal (Section~\ref{sec:numerical}). These results provide theoretical grounding for the tiered ``Provisioned Throughput'' offerings observed across cloud inference markets.

\paragraph{Limitations and Extensions}

Several assumptions merit scrutiny. Our independence assumption (Assumption~\ref{ass:independence}) rules out plausible correlations between willingness-to-pay and time sensitivity; correlated multi-dimensional screening would yield richer predictions about cross-subsidization. The monopolist assumption abstracts from competition, though the separation result derives from buyer preferences (single-crossing in $\gamma$) rather than market power and would constrain prices but not menu structure under competition. We also abstract from queueing dynamics and capacity constraints; incorporating stochastic demand would connect the framework to dynamic pricing and congestion externalities. Finally, our discrete-tier approximation with $K=2$ tiers yields substantial profit gains in favorable regimes (Remark~\ref{rem:few-tiers}), but characterizing optimal mechanisms with continuous tier choices remains open.

Our empirical calibration is limited to a single hardware pair (H100 and B200) with a modest decode-latency differential ($d_{B200}/d_{H100} = 0.68$). As discussed in Section~\ref{sec:profit-comparison}, tiering gains are largest when hardware costs constitute a significant fraction of per-request value and vanish when the premium tier uniformly dominates. Because generational GPU upgrades tend to improve all latency dimensions simultaneously, they compress the cost-performance tradeoff that sustains profitable tiering. Extending the mechanism to higher-WTP markets may require large GPU clusters or architecturally specialized hardware pairings which specifically optimize for inference latency, a prediction testable as inference-specific accelerators (e.g., memory-bandwidth-optimized decode chips) enter the market.

Three extensions are natural: (i)~learning models where buyers discover time preferences through usage, explaining observed tier-switching behavior; (ii)~including specialized inference hardware as the 'premium' hardware; and (iii)~updated latency functions for evolving architectures (speculative decoding, diffusion-based generation), which alter the cost function's form but preserve the throughput--latency tradeoff the separation theorem exploits.

\bibliographystyle{ACM-Reference-Format}
\bibliography{main,newbib}

\appendix

\section{Latency Function Estimation}
\label{app:empirical}

This appendix provides details on the estimation of the latency function~\eqref{eq:latency} used throughout the analysis.

\subsection{Computational Phases of LLM Inference}

LLM inference decomposes into two distinct computational phases with different resource constraints:

\paragraph{Prefill phase (compute-bound).} This phase processes the input context and generates the first output token. The primary computational burden is the self-attention mechanism, which requires $O(n^2)$ operations in the input sequence length $n$. For transformer architectures, attention computation dominates prefill time, particularly for long contexts. In continuous-batching inference systems (e.g., vLLM, TensorRT-LLM), prefill requests are processed individually to avoid interference between compute-bound operations. Consequently, prefill latency depends only on input token count $x$, not on batch size: $L_{\text{prefill}} = P(x) = ax^2 + bx + c$. The quadratic term captures attention computation, the linear term captures FFN and projection layers, and the constant captures fixed overhead (kernel launch, synchronization).

\paragraph{Decode phase (memory-bandwidth-bound).} This phase generates all subsequent tokens autoregressively. Each token generation requires loading model parameters from HBM---the dominant cost---plus accessing the KV-cache for attention. The base decode cost $d$ per output token reflects weight loading and is largely independent of batch size until compute saturation. The batch-dependent terms $eB$ and $fBx$ capture per-batch scheduling overhead and KV-cache attention costs that scale with both batch size and context length. Total decode time scales linearly in output tokens: $L_{\text{decode}} = y \cdot (d + eB + fBx)$.

The combined latency is:
\begin{equation}
    L = P(x) + y \cdot G(B, x) = (ax^2 + bx + c) + y(d + eB + fBx).
\end{equation}

\subsection{Simulation Methodology}

\paragraph{Hardware configurations.} We run Llama3-70B on two GPU generations:
\begin{itemize}
    \item \textbf{H100 (economy tier)}: 80GB HBM3, 3.35 TB/s memory bandwidth, 989 TFLOPS FP16
    \item \textbf{B200 (premium tier)}: 180GB HBM3e, 8 TB/s memory bandwidth, 2250 TFLOPS FP16
\end{itemize}

\paragraph{Parameter ranges.} We vary:
\begin{itemize}
    \item Batch size $B \in \{1, 2, 4, 8, 16, 32, 64\}$
    \item Input context $x \in \{1024, 2048, 4096, 8192\}$ tokens
\end{itemize}

\subsection{Regression Specification}

We fit the parametric form:
\begin{equation}
    L(B, x, y, g) = a_g x^2 + b_g x + c_g + y(d_g + e_g B + f_g Bx)
\end{equation}
via nonlinear least squares, minimizing:
\begin{equation}
    \sum_{(B, x, y) \in \mathcal{D}} \left( L_{\text{simulated}}(B, x, y) - L_{\text{model}}(B, x, y) \right)^2
\end{equation}
where $\mathcal{D}$ is the grid of simulated configurations.

\subsection{Estimated Parameters}

\begin{table}[h]
\centering
\begin{tabular}{lcc}
\toprule
Parameter & H100 ($g_L$) & B200 ($g_H$) \\
\midrule
$a_g$ (prefill quadratic) & $1.077 \times 10^{-9}$ & $3.530 \times 10^{-10}$ \\
$b_g$ (prefill linear) & $2.815 \times 10^{-5}$ & $1.243 \times 10^{-5}$ \\
$c_g$ (prefill constant) & $3.344 \times 10^{-2}$ & $5.210 \times 10^{-2}$ \\
$d_g$ (decode base) & $1.191 \times 10^{-2}$ & $8.147 \times 10^{-2}$ \\
$e_g$ (decode batch overhead) & $3.398 \times 10^{-5}$ & $3.843 \times 10^{-5}$ \\
$f_g$ (decode KV-cache attention) & $5.551 \times 10^{-9}$ & $ 6.848 \times 10^{-9}$ \\
\midrule
$R^2$ & $0.9951$ & $0.9967$ \\
\bottomrule
\end{tabular}
\caption{Estimated latency function parameters. All coefficients are in seconds.}
\label{tab:latency-params}
\end{table}

Both fits achieve $R^2 > 0.99$, indicating the parametric form captures the essential structure of inference latency across the operational range.
\section{Proof of Theorem \ref{thm:separation}}
\label{app:separation}

\begin{proof}
\textbf{Step 1: Incentive constraints.} Fix $w > 0$. For type $\gamma_H$ to prefer $g_H$:
\begin{equation}
    w \cdot q_{g_H}(\gamma_H) - T_{g_H} \geq w \cdot q_{g_L}(\gamma_H) - T_{g_L}.
\end{equation}

For type $\gamma_L$ to prefer $g_L$:
\begin{equation}
    w \cdot q_{g_L}(\gamma_L) - T_{g_L} \geq w \cdot q_{g_H}(\gamma_L) - T_{g_H}.
\end{equation}

\textbf{Step 2: Existence of separating prices.} Combining these constraints, prices satisfying both IC conditions exist if and only if:
\begin{equation}
    w[q_{g_H}(\gamma_L) - q_{g_L}(\gamma_L)] \leq T_{g_H} - T_{g_L} \leq w[q_{g_H}(\gamma_H) - q_{g_L}(\gamma_H)].
\end{equation}
This interval is non-empty precisely when condition~\eqref{eq:separation-condition} holds, which is independent of $w$. For any given $w$, any price gap $T_{g_H} - T_{g_L}$ in this interval achieves separation; in the optimal mechanism (Theorem~\ref{thm:optimal-menu}), the within-tier Mussa-Rosen schedules supply the requisite $w$-dependent transfers.
\end{proof}

\section{Scale-Up Domain Architecture as Tier Differentiation}
\label{app:scaleup}

Hardware tier differentiation extends beyond GPU generation (e.g., H100 vs.\ B200) to include \emph{scale-up domain architecture}---the manner in which multiple GPUs are interconnected within a single inference unit. Modern LLM inference employs \emph{tensor parallelism}: decomposing a single model's computation across multiple tightly-coupled GPUs that operate synchronously on different portions of the model's parameters. This enables strong scaling for individual inference tasks, distinct from the weak scaling achieved by replicating entire models across independent GPU clusters or pipeline parallelism (which spreads layers across devices but does not reduce latency for a single request).

The economic significance lies in the cost-performance tradeoff of interconnect technology. Consider 8 GPUs of identical specification: connected via standard PCIe buses (bandwidth ${\sim}32$ GB/s per link), cross-GPU communication creates severe bottlenecks that preclude effective tensor parallelism, forcing each GPU to operate independently on separate inference requests. Connected via high-bandwidth interconnects such as NVLink switches (bandwidth ${\sim}900$ GB/s bidirectional per link in NVLink 4.0), the same 8 GPUs can act as a coherent computational unit, partitioning the model's attention matrices and linear transformations across devices. For the computationally expensive attention phase---which dominates prefill latency and scales quadratically in input length---this parallelization delivers near-linear speedup: an 8-GPU tensor-parallel configuration can achieve approximately $5\times$ faster time-to-first-token compared to a single GPU processing the same request \citep{davies2025liminal}.

This scale-up capability comes at substantial capital cost: NVLink switches, high-speed interconnect fabrics (NVSwitch, InfiniBand), and the engineering integration required to maintain coherence across dozens of devices represent significant incremental investment beyond the GPUs themselves. Modern scale-up domains range from 8-GPU nodes (common in cloud offerings like AWS p5.48xlarge instances with $8 \times$ H100) to 72-GPU superblocks (Nvidia's DGX SuperPOD NVL72 architecture). Each additional tier of scale-up---8-GPU vs.\ 16-GPU vs.\ 72-GPU configurations---yields progressively lower per-request latency but at higher per-unit-time cost $c_g$, creating a natural hierarchy of hardware tiers differentiated by scale-up capacity rather than solely by GPU generation.

From a pricing perspective, these scale-up tiers provide an additional source of cost-latency tradeoff consistent with the model's structure: sellers invest in expensive interconnect infrastructure to offer low-latency inference that would be uneconomical with weak scaling alone. Buyers with high time preference $\gamma$ value these scale-up-enabled tiers disproportionately, since the latency reduction afforded by tensor parallelism compounds with their $\gamma$ exponent in the production function. Scale-up domain architecture thus exemplifies the \emph{architectural specialization} discussed in Section~\ref{sec:numerical}: tiers differentiated by infrastructure topology rather than GPU generation, creating complementary operating regimes that the separation theorem can exploit.

\section{Derivation of Cost Function Derivatives}
\label{app:cost-derivatives}

This appendix provides the detailed derivation of the first and second derivatives of the per-task cost function $\tilde{c}^*(q)$ used in the proof of Theorem~\ref{thm:cost-convexity}.

\subsection{Setup and Notation}

Recall the key definitions:
\begin{itemize}
    \item Latency function: $L = B \cdot D(x,y) + F(x,y)$ where
    \begin{align}
        D(x,y) &= y(e + fx) \\
        F(x,y) &= ax^2 + bx + c + yd
    \end{align}
    \item Latency slack: $\mu \equiv L_{max} - L$
    \item Production constraint: $x^\alpha y^\beta \mu^\gamma = q$
    \item Numerator function: $N \equiv L_{max} - F(x,y) - \mu$
\end{itemize}

From the production constraint, the required slack to achieve quality $q$ is:
\begin{equation}
    \mu = \left(\frac{q}{x^\alpha y^\beta}\right)^{1/\gamma} = q^{1/\gamma} x^{-\alpha/\gamma} y^{-\beta/\gamma}
    \label{eq:mu-constraint}
\end{equation}

From $L = BD + F$ and $\mu = L_{max} - L$, the batch size is determined by:
\begin{equation}
    B = \frac{L_{max} - \mu - F(x,y)}{D(x,y)} = \frac{N}{D}
    \label{eq:batch-size}
\end{equation}

The per-task cost is:
\begin{equation}
    \tilde{c} = \frac{c_g L}{B} = \frac{c_g L \cdot D}{N}
\end{equation}

Since $L = L_{max} - \mu$:
\begin{equation}
    \tilde{c} = \frac{c_g (L_{max} - \mu) D}{N}
    \label{eq:cost-form}
\end{equation}

\subsection{First Derivative (Marginal Cost)}

We apply the envelope theorem \citep{milgrom2002envelope}: at the optimum $(x^*, y^*)$, the derivative of the minimized cost with respect to $q$ equals the partial derivative of the cost function with respect to $q$, holding $(x, y)$ fixed at their optimal values.

\textbf{Step 1: Compute $\partial \mu / \partial q$.}

From equation~\eqref{eq:mu-constraint}:
\begin{equation}
    \frac{\partial \mu}{\partial q} = \frac{1}{\gamma} q^{1/\gamma - 1} x^{-\alpha/\gamma} y^{-\beta/\gamma} = \frac{1}{\gamma q} \cdot q^{1/\gamma} x^{-\alpha/\gamma} y^{-\beta/\gamma} = \frac{\mu}{\gamma q}
    \label{eq:dmu-dq}
\end{equation}

\textbf{Step 2: Compute $\partial N / \partial q$.}

Since $N = L_{max} - F(x,y) - \mu$ and $F(x,y)$ does not depend on $q$ (holding $(x,y)$ fixed):
\begin{equation}
    \frac{\partial N}{\partial q} = -\frac{\partial \mu}{\partial q} = -\frac{\mu}{\gamma q}
    \label{eq:dN-dq}
\end{equation}

\textbf{Step 3: Apply the quotient rule to $\tilde{c}$.}

From equation~\eqref{eq:cost-form}, with $D$ and $F$ treated as constants (holding $(x,y)$ fixed):
\begin{align}
    \frac{\partial \tilde{c}}{\partial q} &= c_g D \cdot \frac{\partial}{\partial q}\left[\frac{L_{max} - \mu}{N}\right] \\
    &= c_g D \left[ \frac{-\frac{\partial \mu}{\partial q} \cdot N - (L_{max} - \mu) \cdot \frac{\partial N}{\partial q}}{N^2} \right]
\end{align}

Substituting equations~\eqref{eq:dmu-dq} and~\eqref{eq:dN-dq}:
\begin{align}
    \frac{\partial \tilde{c}}{\partial q} &= c_g D \left[ \frac{-\frac{\mu}{\gamma q} \cdot N - (L_{max} - \mu) \cdot \left(-\frac{\mu}{\gamma q}\right)}{N^2} \right] \\
    &= c_g D \left[ \frac{-\frac{\mu N}{\gamma q} + \frac{(L_{max} - \mu)\mu}{\gamma q}}{N^2} \right] \\
    &= \frac{c_g D \mu}{\gamma q N^2} \left[ -N + (L_{max} - \mu) \right]
\end{align}

\textbf{Step 4: Simplify using the definition of $N$.}

Since $N = L_{max} - F(x,y) - \mu$:
\begin{equation}
    (L_{max} - \mu) - N = (L_{max} - \mu) - (L_{max} - F(x,y) - \mu) = F(x,y)
\end{equation}

Therefore:
\begin{equation}
    \boxed{\frac{\partial \tilde{c}}{\partial q} = \frac{c_g D \mu F}{\gamma q N^2}}
    \label{eq:first-derivative}
\end{equation}

Since $c_g > 0$, $D = y(e + fx) > 0$ for $y > 0$ and $e + f > 0$, $\mu > 0$, $F = ax^2 + bx + c + yd > 0$ for $x, y > 0$, $\gamma > 0$, $q > 0$, and $N > 0$ on the feasible region, we have $\partial \tilde{c} / \partial q > 0$. \qed

\subsection{Second Derivative (Convexity)}

We differentiate equation~\eqref{eq:first-derivative} with respect to $q$.

\textbf{Step 1: Rewrite the first derivative.}

\begin{equation}
    \frac{\partial \tilde{c}}{\partial q} = \frac{c_g D F}{\gamma} \cdot \frac{\mu}{q N^2}
\end{equation}

Since $D$ and $F$ are held constant (functions of $(x,y)$ only), we need:
\begin{equation}
    \frac{\partial^2 \tilde{c}}{\partial q^2} = \frac{c_g D F}{\gamma} \cdot \frac{\partial}{\partial q}\left[\frac{\mu}{q N^2}\right]
\end{equation}

\textbf{Step 2: Apply the quotient rule.}

Let $u = \mu$ and $v = q N^2$. Then:
\begin{equation}
    \frac{\partial}{\partial q}\left[\frac{u}{v}\right] = \frac{\frac{\partial u}{\partial q} \cdot v - u \cdot \frac{\partial v}{\partial q}}{v^2}
\end{equation}

We have:
\begin{align}
    \frac{\partial u}{\partial q} &= \frac{\partial \mu}{\partial q} = \frac{\mu}{\gamma q} \\
    \frac{\partial v}{\partial q} &= \frac{\partial}{\partial q}[q N^2] = N^2 + q \cdot 2N \cdot \frac{\partial N}{\partial q} = N^2 + 2qN \cdot \left(-\frac{\mu}{\gamma q}\right) = N^2 - \frac{2\mu N}{\gamma}
\end{align}

\textbf{Step 3: Compute the numerator.}

\begin{align}
    \frac{\partial u}{\partial q} \cdot v - u \cdot \frac{\partial v}{\partial q} &= \frac{\mu}{\gamma q} \cdot q N^2 - \mu \cdot \left(N^2 - \frac{2\mu N}{\gamma}\right) \\
    &= \frac{\mu N^2}{\gamma} - \mu N^2 + \frac{2\mu^2 N}{\gamma} \\
    &= \mu N^2 \left(\frac{1}{\gamma} - 1\right) + \frac{2\mu^2 N}{\gamma} \\
    &= \frac{\mu N^2 (1 - \gamma)}{\gamma} + \frac{2\mu^2 N}{\gamma} \\
    &= \frac{\mu N}{\gamma} \left[ (1-\gamma)N + 2\mu \right]
\end{align}

\textbf{Step 4: Assemble the second derivative.}

\begin{align}
    \frac{\partial}{\partial q}\left[\frac{\mu}{q N^2}\right] &= \frac{\frac{\mu N}{\gamma} \left[ (1-\gamma)N + 2\mu \right]}{q^2 N^4} \\
    &= \frac{\mu \left[ (1-\gamma)N + 2\mu \right]}{\gamma q^2 N^3}
\end{align}

Therefore:
\begin{equation}
    \boxed{\frac{\partial^2 \tilde{c}}{\partial q^2} = \frac{c_g D F \mu \left[ (1-\gamma)N + 2\mu \right]}{\gamma^2 q^2 N^3}}
    \label{eq:second-derivative}
\end{equation}

\textbf{Step 5: Verify positivity.}

All factors in the numerator and denominator are positive on the feasible region:
\begin{itemize}
    \item $c_g > 0$ (GPU cost rate)
    \item $D = y(e + fx) > 0$ for $x, y > 0$ (since $e + f > 0$)
    \item $F = ax^2 + bx + c + yd > 0$ for $x, y > 0$ (since $c, d > 0$)
    \item $\mu > 0$ (positive latency slack required for positive quality)
    \item $(1-\gamma) > 0$ since $\gamma < 1$ (by the constraint $\alpha + \beta + \gamma < 1$)
    \item $N > 0$ (required for positive batch size; see equation~\eqref{eq:batch-size})
    \item $\gamma > 0$, $q > 0$
\end{itemize}

Therefore $\partial^2 \tilde{c} / \partial q^2 > 0$, establishing strict convexity of $\tilde{c}^*(q)$. \qed

\section{Proof of Proposition \ref{prop:optimal-transfers}}
\label{app:proof-transfers}
\begin{proof}
By the envelope theorem, for IC to hold with equality at adjacent types:
\begin{equation}
    \frac{dU^*(w)}{dw} = q^*(w)
\end{equation}
where $U^*(w) = w \cdot q^*(w) - T^*(w)$ is the equilibrium utility.

Integrating from the exclusion threshold where $U^*(\underline{w}) = 0$:
\begin{equation}
    U^*(w) = \int_{\underline{w}}^{w} q^*(k) \, dk
\end{equation}

Solving for transfers:
\begin{equation}
    T^*(w) = w \cdot q^*(w) - U^*(w) = w \cdot q^*(w) - \int_{\underline{w}}^{w} q^*(k) \, dk
\end{equation}
\end{proof}

\section{Proof of Proposition \ref{prop:scale-ic}}
\label{app:proof-scale-ic}
\begin{proof}
\textbf{Downward deviations.} Type $(w, s)$ reporting $\tilde{s} < s$ receives utility $\tilde{s} \cdot \tilde{u}^*(w) < s \cdot \tilde{u}^*(w)$ since $\tilde{u}^*(w) \geq 0$.

\textbf{Upward deviations.} Type $(w, s)$ reporting $(\tilde{w}, \tilde{s})$ with $\tilde{s} > s$ receives $\tilde{s}$ allocations but values only $s$ tasks. The optimal deviation sets $\tilde{w}^* = ws/\tilde{s}$ (from the envelope condition), yielding payoff $\tilde{s} \cdot \tilde{u}^*(ws/\tilde{s})$. Since $\tilde{u}(w) = \int_{\underline{w}}^{w} q(k) , dk$ is convex with $\tilde{u}(\underline{w}) = 0$ (extending by zero below $\underline{w}$), convexity gives $\tilde{u}(ws/\tilde{s}) \leq (s/\tilde{s}) \cdot \tilde{u}^*(w)$. Hence the deviation payoff is at most $s \cdot \tilde{u}^*(w)$, the truthful utility.
\end{proof}

\section{Proof of Proposition~\ref{prop:partition-existence}}
\label{app:proof-partition}

\begin{proof}
\textbf{Step 1: Monotone selection.} By Assumption~\ref{ass:quality-separation}, the quality gain $\Delta q_k(\gamma) \equiv q_{g_{k+1}}(\gamma) - q_{g_k}(\gamma)$ is strictly increasing in $\gamma$. Consider type $\gamma'$ who prefers $g_{k+1}$ over $g_k$, i.e., the quality gain justifies the price premium:
\begin{equation}
    w \cdot \Delta q_k(\gamma') \geq T_{g_{k+1}} - T_{g_k}.
\end{equation}
For any $\gamma'' > \gamma'$, we have $\Delta q_k(\gamma'') > \Delta q_k(\gamma')$ by Assumption~\ref{ass:quality-separation}, so:
\begin{equation}
    w \cdot \Delta q_k(\gamma'') > w \cdot \Delta q_k(\gamma') \geq T_{g_{k+1}} - T_{g_k}.
\end{equation}
Hence $\gamma''$ also strictly prefers $g_{k+1}$ over $g_k$. This establishes monotone selection: if any type prefers the higher tier, all higher types do as well.

\textbf{Step 2: Existence of cutoffs.} Fix tier prices $(T_{g_1}, \ldots, T_{g_K})$ with $T_{g_1} < T_{g_2} < \cdots < T_{g_K}$. Define the indifference function for adjacent tiers:
\begin{equation}
    \Phi_k(\gamma) \equiv w \cdot \Delta q_k(\gamma) - (T_{g_{k+1}} - T_{g_k}).
\end{equation}

By Assumption~\ref{ass:quality-separation}, $\Phi_k(\gamma)$ is strictly increasing in $\gamma$. At $\gamma = \underline{\gamma}$, $\Phi_k(\underline{\gamma})$ may be negative (low-$\gamma$ types prefer economy tier). At $\gamma = \bar{\gamma}$, $\Phi_k(\bar{\gamma})$ may be positive (high-$\gamma$ types prefer premium tier). By the intermediate value theorem, there exists a unique $\gamma_k^* \in [\underline{\gamma}, \bar{\gamma}]$ such that $\Phi_k(\gamma_k^*) = 0$, with:
\begin{itemize}
    \item Types $\gamma < \gamma_k^*$ have $\Phi_k(\gamma) < 0$: prefer $g_k$;
    \item Types $\gamma > \gamma_k^*$ have $\Phi_k(\gamma) > 0$: prefer $g_{k+1}$;
    \item Type $\gamma_k^*$ is indifferent.
\end{itemize}

\textbf{Step 3: Partition structure.} Applying Step 2 iteratively for $k = 1, \ldots, K-1$ yields cutoffs $\gamma_1^* < \gamma_2^* < \cdots < \gamma_{K-1}^*$. The ordering follows from monotone selection: if $\gamma_k^* \geq \gamma_{k+1}^*$, then types in $[\gamma_{k+1}^*, \gamma_k^*]$ would simultaneously prefer $g_k$ over $g_{k-1}$ and $g_{k+1}$ over $g_k$, but also prefer $g_{k+2}$ over $g_{k+1}$---contradicting the definition of $\gamma_{k+1}^*$ as the indifference point between $g_{k+1}$ and $g_{k+2}$.
\end{proof}

\section{Proof of Lemma~\ref{lem:avg-cost-properties}}
\label{app:proof-avg-cost}

\begin{proof}
\textbf{Property 1 (Strict convexity):} A convex combination (expectation) of strictly convex functions is strictly convex. For any $q_1 \neq q_2$ and $\lambda \in (0,1)$:
\begin{equation}
    \tilde{c}^*(\lambda q_1 + (1-\lambda) q_2; g_k, \gamma) < \lambda \tilde{c}^*(q_1; g_k, \gamma) + (1-\lambda) \tilde{c}^*(q_2; g_k, \gamma)
\end{equation}
holds for each $\gamma$ by Theorem~\ref{thm:cost-convexity}. Integrating over $\gamma \in \Gamma_k$ preserves the strict inequality since the integrand is strictly negative on a set of positive measure.

\textbf{Property 2 (Boundary behavior):} For each $\gamma$, we have $\tilde{c}^{*\prime}(0; g_k, \gamma) = 0$ and $\tilde{c}^{*\prime}(q; g_k, \gamma) \to \infty$ as $q \to \bar{q}(g_k, \gamma)$ by Theorem~\ref{thm:cost-convexity}. Taking expectations:
\begin{equation}
    \tilde{c}_{\mathrm{avg}}'(0; g_k, \boldsymbol{\gamma}_k) = \frac{1}{\Pr(\Gamma_k)} \int_{\gamma_{k-1}}^{\gamma_k} \tilde{c}^{*\prime}(0; g_k, \gamma) \, f_\gamma(\gamma) \, d\gamma = 0.
\end{equation}
For the upper bound, as $q$ approaches the minimum feasibility bound across $\gamma \in \Gamma_k$, at least some types have $\tilde{c}^{*\prime}(q; g_k, \gamma) \to \infty$, so the average diverges.

\textbf{Property 3 (Continuity in cutoffs):} Fix $q$ in the interior of the feasible domain. The function $\tilde{c}^*(q; g_k, \gamma)$ is continuous in $\gamma$ (by continuity of the cost-minimization problem in parameters). The averaged cost is:
\begin{equation}
    \tilde{c}_{\mathrm{avg}}(q; g_k, \boldsymbol{\gamma}_k) = \frac{\int_{\gamma_{k-1}}^{\gamma_k} \tilde{c}^*(q; g_k, \gamma) \, f_\gamma(\gamma) \, d\gamma}{F_\gamma(\gamma_k) - F_\gamma(\gamma_{k-1})}.
\end{equation}
Both numerator and denominator are continuous in $(\gamma_{k-1}, \gamma_k)$ by standard results on parameter-dependent integrals. The denominator is bounded away from zero for $\gamma_k > \gamma_{k-1}$ (and we consider only valid partitions). Hence $\tilde{c}_{\mathrm{avg}}$ is continuous in the cutoffs.
\end{proof}

\section{Proof of Theorem~\ref{thm:optimal-cutoffs}: Necessary Condition}
\label{app:proof-cutoffs-a}

\begin{proof}
The seller's total expected profit is:
\begin{equation}
    \Pi(\boldsymbol{\gamma}) = \sum_{k=1}^{K} \int_{\gamma_{k-1}}^{\gamma_k} \int_{\underline{w}}^{\bar{w}} \left[ T_k^*(w; \boldsymbol{\gamma}) - \tilde{c}^*(q_k^*(w; \boldsymbol{\gamma}); g_k, \gamma) \right] f_w(w) \, f_\gamma(\gamma) \, dw \, d\gamma,
\end{equation}
where the within-tier mechanisms $(q_k^*, T_k^*)$ are optimal for the Mussa-Rosen problem with cost $\tilde{c}_{\mathrm{avg}}(\cdot; g_k, \boldsymbol{\gamma}_k)$, and hence depend on the cutoffs $\boldsymbol{\gamma}$.

For segment $k$, write the contribution as:
\begin{equation}
    V_k(\boldsymbol{\gamma}) = \int_{\gamma_{k-1}}^{\gamma_k} \int_{\underline{w}}^{\bar{w}} \left[ T_k^*(w) - \tilde{c}^*(q_k^*(w); g_k, \gamma) \right] f_w(w) \, f_\gamma(\gamma) \, dw \, d\gamma.
\end{equation}

A marginal increase $d\gamma_k > 0$ has two effects:

\emph{(i) Boundary effect.} By Leibniz's rule:
\begin{equation}
    \left. \frac{\partial V_k}{\partial \gamma_k} \right|_{\text{boundary}} = f_\gamma(\gamma_k) \int_{\underline{w}}^{\bar{w}} \left[ T_k^*(w) - \tilde{c}^*(q_k^*(w); g_k, \gamma_k) \right] f_w(w) \, dw = f_\gamma(\gamma_k) \cdot \mathbb{E}_w[\Pi_{g_k}(w, \gamma_k)].
\end{equation}

\emph{(ii) Cost-averaging effect.} Changing $\gamma_k$ alters $\tilde{c}_{\mathrm{avg}}$, which shifts the within-tier mechanism. By the envelope theorem applied to the within-tier optimization, the response of $(q_k^*, T_k^*)$ to changes in the cost function does not contribute at the margin (they are at the optimum). However, the direct change in $\tilde{c}_{\mathrm{avg}}$ affects the profit integrand. We have:
\begin{equation}
    \frac{\partial \tilde{c}_{\mathrm{avg}}(q; g_k, \boldsymbol{\gamma}_k)}{\partial \gamma_k} = \frac{f_\gamma(\gamma_k)}{\Pr(\Gamma_k)} \left[ \tilde{c}^*(q; g_k, \gamma_k) - \tilde{c}_{\mathrm{avg}}(q; g_k, \boldsymbol{\gamma}_k) \right].
\end{equation}

The cost-averaging effect contributes:
\begin{equation}
    \left. \frac{\partial V_k}{\partial \gamma_k} \right|_{\text{avg}} = -\int_{\gamma_{k-1}}^{\gamma_k} \int_{\underline{w}}^{\bar{w}} \frac{f_\gamma(\gamma_k)}{\Pr(\Gamma_k)} \left[ \tilde{c}^*(q_k^*; g_k, \gamma_k) - \tilde{c}_{\mathrm{avg}}(q_k^*; g_k) \right] f_w(w) \, f_\gamma(\gamma) \, dw \, d\gamma.
\end{equation}

Combining (i) and (ii), the $\tilde{c}_{\mathrm{avg}}$ terms cancel, yielding:
\begin{equation}
    \frac{\partial V_k}{\partial \gamma_k} = f_\gamma(\gamma_k) \cdot \mathbb{E}_w[\Pi_{g_k}(w, \gamma_k)].
\end{equation}

An analogous calculation for segment $k+1$ (whose lower bound is $\gamma_k$) gives:
\begin{equation}
    \frac{\partial V_{k+1}}{\partial \gamma_k} = -f_\gamma(\gamma_k) \cdot \mathbb{E}_w[\Pi_{g_{k+1}}(w, \gamma_k)].
\end{equation}

The total derivative is:
\begin{equation}
    \frac{\partial \Pi}{\partial \gamma_k} = f_\gamma(\gamma_k) \cdot \mathbb{E}_w[\Pi_{g_k}(w, \gamma_k) - \Pi_{g_{k+1}}(w, \gamma_k)].
\end{equation}
Setting $\partial \Pi / \partial \gamma_k = 0$ yields the indifference condition~\eqref{eq:indifference}.
\end{proof}

\section{Proof of Theorem~\ref{thm:optimal-cutoffs}: Existence and Sufficiency}
\label{app:proof-cutoffs-bc}

\begin{proof}
\textbf{Part (b): Existence via Brouwer's fixed point theorem.}

Define the cutoff space $\mathcal{C} = [\underline{\gamma}, \bar{\gamma}]^{K-1}$, which is compact and convex. For degenerate cutoff vectors with $\gamma_{k-1} = \gamma_k$, the averaged cost is extended continuously via $\tilde{c}_{\mathrm{avg}}(q; g_k, (\gamma_k, \gamma_k)) \equiv \tilde{c}^*(q; g_k, \gamma_k)$, ensuring the within-tier Mussa-Rosen problem remains well-defined.

For any $\boldsymbol{\gamma} \in \mathcal{C}$, the within-tier mechanisms $(q_k^*, T_k^*)$ are uniquely determined by the Mussa-Rosen solution with cost $\tilde{c}_{\mathrm{avg}}(\cdot; g_k, \boldsymbol{\gamma}_k)$. Define the profit-difference function:
\begin{equation}
    \Delta_k(\gamma; \boldsymbol{\gamma}) \equiv \mathbb{E}_w[\Pi_{g_{k+1}}(w, \gamma; \boldsymbol{\gamma}) - \Pi_{g_k}(w, \gamma; \boldsymbol{\gamma})],
\end{equation}
where we make explicit the dependence on the full cutoff vector $\boldsymbol{\gamma}$ through the within-tier mechanisms.

For each $k$, define the best-response cutoff:
\begin{equation}
    \hat{\gamma}_k(\boldsymbol{\gamma}) =
    \begin{cases}
        \underline{\gamma} & \text{if } \Delta_k(\underline{\gamma}; \boldsymbol{\gamma}) \geq 0, \\
        \bar{\gamma} & \text{if } \Delta_k(\bar{\gamma}; \boldsymbol{\gamma}) \leq 0, \\
        \gamma \text{ s.t. } \Delta_k(\gamma; \boldsymbol{\gamma}) = 0 & \text{otherwise}.
    \end{cases}
\end{equation}

By Assumption~\ref{ass:quality-separation} and the argument in Proposition~\ref{prop:partition-existence}, $\Delta_k(\gamma; \boldsymbol{\gamma})$ is strictly increasing in $\gamma$ for fixed $\boldsymbol{\gamma}$, so the root (if interior) is unique. Since each $\hat{\gamma}_k(\boldsymbol{\gamma}) \in [\underline{\gamma}, \bar{\gamma}]$ by construction, the mapping $\mathcal{T}: \mathcal{C} \to \mathcal{C}$ defined by $\mathcal{T}(\boldsymbol{\gamma}) = (\hat{\gamma}_1(\boldsymbol{\gamma}), \ldots, \hat{\gamma}_{K-1}(\boldsymbol{\gamma}))$ maps $\mathcal{C}$ to itself.

\emph{Continuity of $\mathcal{T}$:} By Lemma~\ref{lem:avg-cost-properties} (Property 3), $\tilde{c}_{\mathrm{avg}}$ is continuous in the cutoffs. The Mussa-Rosen solution $(q_k^*, T_k^*)$ depends continuously on the cost function (by standard results on parametric optimization), so the profit functions $\Pi_g(w, \gamma; \boldsymbol{\gamma})$ are continuous in $\boldsymbol{\gamma}$. The expectation $\Delta_k(\gamma; \boldsymbol{\gamma})$ inherits this continuity. The root of a strictly monotone continuous function varies continuously with parameters, so $\hat{\gamma}_k(\boldsymbol{\gamma})$ is continuous. Hence $\mathcal{T}$ is continuous.

By Brouwer's fixed point theorem, $\mathcal{T}$ has a fixed point $\boldsymbol{\gamma}^* = \mathcal{T}(\boldsymbol{\gamma}^*)$. At this fixed point, either $\gamma_k^*$ is at a boundary (and the FOC need not hold) or $\Delta_k(\gamma_k^*; \boldsymbol{\gamma}^*) = 0$, which is the indifference condition~\eqref{eq:indifference}.

\emph{Ordering at the fixed point.} For $K = 2$, the ordering constraint is vacuous ($\mathcal{C} = [\underline{\gamma}, \bar{\gamma}]$). For $K \geq 3$, a fixed point $\boldsymbol{\gamma}^*$ may have $\gamma_k^* \geq \gamma_{k+1}^*$ for some $k$, corresponding to a degenerate partition in which tier $g_{k+1}$ serves an empty segment and is effectively inactive. Such solutions reduce to partitions with fewer than $K$ active tiers. Non-degenerate fixed points with all tiers active satisfy $\gamma_1^* < \cdots < \gamma_{K-1}^*$ and the indifference conditions~\eqref{eq:indifference} hold at each boundary.

\textbf{Part (c): Second-order condition.}

The Hessian of $\Pi$ with respect to $\boldsymbol{\gamma}$ at an interior solution has diagonal entries:
\begin{equation}
    \frac{\partial^2 \Pi}{\partial \gamma_k^2} = f_\gamma'(\gamma_k) \cdot \mathbb{E}_w[\Pi_{g_k} - \Pi_{g_{k+1}}] + f_\gamma(\gamma_k) \cdot \frac{\partial}{\partial \gamma_k} \mathbb{E}_w[\Pi_{g_k} - \Pi_{g_{k+1}}].
\end{equation}

At the FOC where $\mathbb{E}_w[\Pi_{g_k} - \Pi_{g_{k+1}}] = 0$, this simplifies to:
\begin{equation}
    \frac{\partial^2 \Pi}{\partial \gamma_k^2} = f_\gamma(\gamma_k) \cdot \frac{\partial}{\partial \gamma_k} \mathbb{E}_w[\Pi_{g_k}(w, \gamma_k) - \Pi_{g_{k+1}}(w, \gamma_k)].
\end{equation}

Since the within-tier mechanisms $(q_k^*, T_k^*)$ do not depend on $\gamma$ (only on the cutoffs through $\tilde{c}_{\mathrm{avg}}$), this derivative reduces to:
\begin{equation}
    \frac{\partial}{\partial \gamma} \mathbb{E}_w[\Pi_{g_k} - \Pi_{g_{k+1}}] = \mathbb{E}_w \left[ \frac{\partial \tilde{c}^*(q_{k+1}^*; g_{k+1}, \gamma)}{\partial \gamma} - \frac{\partial \tilde{c}^*(q_k^*; g_k, \gamma)}{\partial \gamma} \right].
\end{equation}

For local concavity (SOC for a maximum), we require $\partial^2 \Pi / \partial \gamma_k^2 \leq 0$, which holds when condition~\eqref{eq:SOC} is satisfied.
\end{proof}

\section{AI Use Disclosure}

We utilized Large Language Models (Gemini 3 Pro/Claude Opus 4.6) to act as an adversarial critic during the drafting process, specifically to stress-test the economic logic of the separation theorem and to identify potential gaps in the related work. The tools were also used for proof-verification assistance and stylistic refinement. All final text, mathematical derivations, and citations were verified and authored by us.

\end{document}